\documentclass[11pt,a4paper]{amsart}
\usepackage[dvips,bookmarks,colorlinks=true,linkcolor=blue,urlcolor=red]{hyperref}
\numberwithin{equation}{section}
\newcommand{\D}{\displaystyle}

\newcommand{\car}{\mathbf{1}}
\newcommand{\R}{{\mathbb R}}

\newcommand{\Z}{{\mathbb Z}}
\newcommand{\N}{{\mathbb N}}
\newcommand{\C}{{\mathbb C}}

\newcommand\E{\mathbb E}
\newcommand{\esp}{{\mathbb E}}

\renewcommand\P{\mathbb P}
\newcommand{\dist}{{\rm dist}\,}

\newcommand{\vers}{\operatornamewithlimits{\to}}
\theoremstyle{plain}
\newtheorem{Th}{Theorem}[section]
\newtheorem{Le}{Lemma}[section]
\newtheorem{Pro}{Proposition}[section]

\theoremstyle{definition}
\newtheorem{Rem}{Remark}[section]
\title{Decorrelation estimates for the eigenlevels of the discrete
  Anderson model in the localized regime}
\author{Fr{\'e}d{\'e}ric Klopp}
\address[Fr{\'e}d{\'e}ric Klopp]{LAGA, U.M.R. 7539 C.N.R.S, Institut Galil{\'e}e,
  Universit{\'e} Paris-Nord, 99 Avenue J.-B. Cl{\'e}ment, F-93430
  Villetaneuse, France\ et \ Institut Universitaire de France}
\email{\href{mailto:klopp@math.univ-paris13.fr}{klopp@math.univ-paris13.fr}}
\begin{document}
\thanks{The author is supported by the grant ANR-08-BLAN-0261-01}
\begin{abstract}
  The purpose of the present work is to establish decorrelation
  estimates for the eigenvalues of the discrete Anderson model
  localized near two distinct energies inside the localization
  region. In dimension one, we prove these estimates at all
  energies. In higher dimensions, the energies are required to be
  sufficiently far apart from each other. As a consequence of these
  decorrelation estimates, we obtain the independence of the limits of
  the local level statistics at two distinct energies.
  \vskip.5cm\noindent \textsc{R{\'e}sum{\'e}.} Dans ce travail, nous
  {\'e}tablissons des in{\'e}galit{\'e}s de d{\'e}corr{\'e}lation pour les valeurs propres
  proches de deux {\'e}nergies distinctes. En dimension 1, nous d{\'e}montrons
  que ces in{\'e}galit{\'e}s sont vraies quel que soit le choix de ces deux
  {\'e}nergies. En dimension sup{\'e}rieure, il nous faut supposer que les
  deux {\'e}nergies sont suffisamment {\'e}loign{\'e}es l'une de l'autre. Comme
  cons{\'e}quence de ces in{\'e}galit{\'e}s de d{\'e}corr{\'e}lation, nous d{\'e}montrons que
  les limites des statistiques locales des valeurs propres sont
  ind{\'e}pendantes pour deux {\'e}nergies distinctes.
\end{abstract}
\maketitle
\section{Introduction}
\label{intro}
On $\ell^2(\Z^d)$, consider the random Anderson model
\begin{equation*}
  H_{\omega}=-\Delta+V_\omega
\end{equation*}
where $-\Delta$ is the free discrete Laplace operator
\begin{equation}
  \label{eq:29}
  (-\Delta u)_n=\sum_{|m-n|=1}u_m\quad\text{ for }
  u=(u_n)_{n\in\Z^d}\in\ell^2(\Z^d)
\end{equation}
and $V_\omega$ is the random potential
\begin{equation}
  \label{eq:54}
  (V_\omega u)_n=\omega_n u_n\quad\text{ for }u=(u_n)_{n\in\Z^d}
  \in\ell^2(\Z^d).
\end{equation}
We assume that the random variables $(\omega_n)_{n\in\Z^d}$ are
independent identically distributed and that their common distribution
admits a compactly supported bounded density, say $g$.\\
It is then well known (see e.g.~\cite{MR2509110}) that
\begin{itemize}
\item let $\Sigma:=[-2d,2d]+$supp$\,g$ and $S_-$ and $S_+$ be the
  infimum and supremum of $\Sigma$; for almost every
  $\omega=(\omega_n)_{n\in\Z^d}$, the spectrum of $H_{\omega}$ is
  equal to $\Sigma$;
\item for some $S_-<s_-\leq s_+<S_+$, the intervals $I_-=[S_-,s_-)$
  and $I_+=(s_+,S_+]$ are contained in the region of localization for
  $H_\omega$ i.e. the region of $\Sigma$ where the finite volume
  fractional moment criteria of~\cite{MR2002h:82051} are verified for
  restrictions of $H_\omega$ to sufficiently large cubes (see also
  Proposition~\ref{pro:1}). In particular, $I:=I_-\cup I_+$ contains
  only pure point spectrum associated to exponentially decaying
  eigenfunctions; for the precise meaning of the region of
  localization, we refer to section~\ref{sec:localized-regime}; if the
  disorder is sufficiently large or if the dimension $d=1$ then, one
  can pick $I=\Sigma$;
\item there exists a bounded density of states, say
  $\lambda\mapsto\nu(E)$, such that, for any continuous function
  $\varphi:\ \R\to\R$, one has
  \begin{equation}
    \label{eq:10}
    \int_\R\varphi(E)\nu(E)dE=
    \mathbb{E}(\langle\delta_0,\varphi(H_\omega)\delta_0\rangle).
  \end{equation}
  Here, and in the sequel, $\mathbb{E}(\cdot)$ denotes the expectation
  with respect to the random parameters, and $\P(\cdot)$ the
  probability measure they induce.\\
  Let $N$ be the integrated density of states of $H_\omega$ i.e. $N$
  is the distribution function of the measure $\nu(E)dE$. The function
  $\nu$ is only defined $E$-almost everywhere. In the sequel, when we
  speak of $\nu(E)$ for some $E$, we mean that the non decreasing
  function $N$ is differentiable at $E$ and that $\nu(E)$ is its
  derivative at $E$.
\end{itemize}
\subsection{The results}
\label{sec:results}
For $L\in\N$, let $\Lambda=\Lambda_L=[-L,L]^d$ be a large box and
$N:=\#\Lambda_L=(2L+1)^d$ be its cardinality. Let
$H_{\omega}(\Lambda)$ be the operator $H_\omega$ restricted to
$\Lambda$ with periodic boundary conditions. The notation
$|\Lambda|\to+\infty$ is a shorthand for considering
$\Lambda=\Lambda_L$ in the limit $L\to+\infty$. Let us denote the
eigenvalues of $H_{\omega}(\Lambda)$ ordered increasingly and repeated
according to multiplicity by $E_1(\omega,\Lambda)\leq
E_2(\omega,\Lambda)\leq \cdots\leq E_N(\omega,\Lambda)$.
\par Let $E$ be an energy in $I$ such that $\nu(E)>0$.  The local
level statistics near $E$ is the point process defined by
\begin{equation}
  \label{eq:13}
  \Xi(\xi,E,\omega,\Lambda) = 
  \sum_{n=1}^N \delta_{\xi_n(E,\omega,\Lambda)}(\xi)
\end{equation}
where
\begin{equation}
  \label{eq:11}
  \xi_n(E,\omega,\Lambda)=|\Lambda|\,\nu(E)\,(E_n(\omega,\Lambda)-
  E),\quad 1\leq n\leq N.
\end{equation}
One of the most striking results describing the localization regime
for the Anderson model is
\begin{Th}[\cite{MR97d:82046}]
  \label{thr:3}
  Assume that $E\in I$ be such that $\nu(E)>0$.\\
  When $|\Lambda|\to+\infty$, the point process
  $\Xi(\cdot,E,\omega,\Lambda)$ converges weakly to a Poisson process
  on $\R$ with intensity the Lebesgue measure i.e. for $(U_j)_{1\leq
    j\leq J}$, $U_j\subset\R$ bounded measurable and $U_{j'}\cap
  U_j=\emptyset$ if $j\not=j'$ and $(k_j)_{1\leq j\leq J}\in\N^J$, one
  has
  \begin{equation*}
    \mathbb{P}\left(\left\{\omega;\
        \begin{cases}
          &\#\{j;\xi_n(E,\omega,\Lambda)\in U_1\}=k_1\\
          &\quad\vdots\hskip3cm\vdots\\
          &\#\{j;\xi_n(E,\omega,\Lambda)\in U_J\}=k_J
        \end{cases}
      \right\}
    \right)\vers_{\Lambda\to\Z^d}\prod_{j=1}^J
    e^{-|U_j|}\frac{|U_j|^{k_j}}{k_j!}.
  \end{equation*}
\end{Th}
\noindent An analogue of Theorem~\ref{thr:3} was first proved
in~\cite{MR673165} for a different one-dimensional random operator.\\
Once Theorem~\ref{thr:3} is known, a natural question arises:
\begin{itemize}
\item for $E\not=E'$, are the limits of $\Xi(\xi,E,\omega,\Lambda)$
  and $\Xi(\xi,E',\omega,\Lambda)$ stochastically independent?
\end{itemize}
This question has arisen and has been answered for other types of
random operators like random matrices (see e.g.~\cite{MR2129906}); in
this case, the local statistics are not Poissonian.\\
For the Anderson model, this question has been open (see
e.g.~\cite{Mi:08,MR2280381}) and to the best of our knowledge, the
present paper is the first to bring an answer. The conjecture is also
open for the continuous Anderson model and random CMV matrices where
the local statistics have also been proved to be Poissonian (see
e.g.~\cite{CGK:10,Ge-Kl:10,MR2280381,MR2220032}).\\
The main result of the present paper is
\begin{Th}
  \label{thr:4}
  Assume that the dimension $d=1$. Pick $E\in I$ and $E'\in I$
  such that $E\not=E'$, $\nu(E)>0$ and $\nu(E')>0$.\\
  When $|\Lambda|\to+\infty$, the point processes
  $\Xi(E,\omega,\Lambda)$ and $\Xi(E',\omega,\Lambda)$, defined
  in~\eqref{eq:13}, converge weakly respectively to two independent
  Poisson processes on $\R$ with intensity the Lebesgue measure. That
  is, for $(U_j)_{1\leq j\leq J}$, $U_j\subset\R$ bounded measurable
  and $U_{j'}\cap U_j=\emptyset$ if $j\not=j'$ and $(k_j)_{1\leq j\leq
    J}\in\N^J$ and $(U'_j)_{1\leq j\leq J'}$, $U'_j\subset\R$ bounded
  measurable and $U'_{j'}\cap U'_j=\emptyset$ if $j\not=j'$ and
  $(k'_j)_{1\leq j\leq J}\in\N^{J'}$ one has
  \begin{equation}
    \label{eq:18}
    \mathbb{P}\left(\left\{\omega;\
        \begin{split}
          &\#\{j;\xi_n(E,\omega,\Lambda)\in U_1\}=k_1\\
          &\quad\vdots\hskip3cm\vdots\\
          &\#\{j;\xi_n(E,\omega,\Lambda)\in U_J\}=k_J\\
          &\#\{j;\xi_n(E',\omega,\Lambda)\in U'_1\}=k'_1\\
          &\quad\vdots\hskip3cm\vdots\\
          &\#\{j;\xi_n(E',\omega,\Lambda)\in U_{J'}\}=k_{J'}
        \end{split}
      \right\} \right)\vers_{\Lambda\to\Z^d}\prod_{j=1}^J
    e^{-|U_j|}\frac{|U_j|^{k_j}}{k_j!}.\prod_{j=1}^{J'}
    e^{-|U_{j'}|}\frac{|U_{j'}|^{k_{j'}}}{k_{j'}!}.
  \end{equation}
\end{Th}
\noindent When $d\geq2$, we also prove
\begin{Th}
  \label{thr:18}
  Assume that $d$ is arbitrary. Pick $E\in I$ and $E'\in I$
  such that $|E-E'|>2d$, $\nu(E)>0$ and $\nu(E')>0$.\\
  When $|\Lambda|\to+\infty$, the point processes
  $\Xi(E,\omega,\Lambda)$ and $\Xi(E',\omega,\Lambda)$, defined
  in~\eqref{eq:13}, converge weakly respectively to two independent
  Poisson processes on $\R$ with intensity the Lebesgue measure.
\end{Th}
\noindent In section~\ref{sec:proofs-theorems}, we show that
Theorems~\ref{thr:4} and~\ref{thr:18} follow from Theorem~\ref{thr:3}
and the decorrelation estimates that we present now. They are the main
technical results of the present paper.
\begin{Le}
  \label{le:1}
  Assume $d=1$ and pick $\beta\in(1/2,1)$. For $\alpha\in(0,1)$ and
  $\{E,E'\}\subset I$ s.t. $E\not=E'$, for any $c>0$, there exists
  $C>0$ such that, for $L\geq3$ and $c L^\alpha\leq \ell\leq
  L^\alpha/c$, one has
  \begin{equation}
    \label{eq:8}
    \P\left(\left\{
        \begin{matrix}
          \sigma(H_\omega(\Lambda_\ell))\cap
          (E+L^{-d}(-1,1))\not=\emptyset,\\
          \sigma(H_\omega(\Lambda_\ell))\cap
          (E'+L^{-d}(-1,1))\not=\emptyset
        \end{matrix}
      \right\}\right)\leq C(\ell/L)^{2d}e^{(\log L)^\beta}.
  \end{equation}
\end{Le}
\noindent This lemma shows that, up to sub-polynomial errors, the
probability to obtain simultaneously an eigenvalue near $E$ and
another one near $E'$ is bounded by the product of the estimates given
for each of these events by Wegner's estimate (see
section~\ref{sec:wegn-minami-estim}). In this sense,~\eqref{eq:8} is
similar to Minami's estimate for two distinct energies.\\
Lemma~\ref{le:1} proves a result conjectured in~\cite{Mi:08,MR2280381}
in dimension 1.
\par In arbitrary dimension, we prove~\eqref{eq:8}, actually a
somewhat stronger estimate, only when the two energies $E$ and $E'$
are sufficiently far apart.
\begin{Le}
  \label{le:9}
  Assume $d$ is arbitrary. Pick $\beta\in(1/2,1)$. For
  $\alpha\in(0,1)$ and $\{E,E'\}\subset I$ s.t. $|E-E'|>2d$, for any
  $c>0$, there exists $C>0$ such that, for $L\geq3$ and $c
  L^\alpha\leq \ell\leq L^\alpha/c$, one has
  \begin{equation}
    \label{eq:16}
    \P\left(\left\{
        \begin{matrix}
          \sigma(H_\omega(\Lambda_\ell))\cap
          (E+L^{-d}(-1,1))\not=\emptyset,\\
          \sigma(H_\omega(\Lambda_\ell))\cap
          (E'+L^{-d}(-1,1))\not=\emptyset
        \end{matrix}
      \right\}\right)\leq C(\ell/L)^{2d}(\log L)^C.
  \end{equation}
\end{Le}
\noindent This e.g. proves the independence of the processes for
energies in opposite edges of the almost sure spectrum.\\
The estimate~\eqref{eq:16} in Lemma~\ref{le:9} is somewhat stronger
than~\eqref{eq:8}; one can obtain an analogous estimate in dimension 1
if one restricts oneself to energies $E$ and $E'$ such that $E-E'$
does not belong to some set of measure $0$ (see Lemma~\ref{le:14} in
Remark~\ref{rem:3} at the end of section~\ref{sec:proof-lemma}).
\begin{Rem}
  \label{rem:2}
  As the proof of Theorems~\ref{thr:4} and~\ref{thr:18} shows, the
  estimates~\eqref{eq:8} are~\eqref{eq:16} are stronger than what it
  needed. It suffices to show that the probabilities in~\eqref{eq:8}
  are~\eqref{eq:16} are $o((\ell/L)^d)$.\\
  In~\cite{Ge-Kl:10} (see also~\cite{Ge-Kl:10b}), the authors provide
  another proof of Theorems~\ref{thr:3} and of Theorems~\ref{thr:4}
  and~\ref{thr:18} under the assumption that the probabilities
  in~\eqref{eq:8} are~\eqref{eq:16} are $o((\ell/L)^d)$. The analysis
  done in~\cite{Ge-Kl:10b} deals with both discrete and continuous
  models. It yields a stronger version of Theorem~\ref{thr:3} and
  Theorems~\ref{thr:4} and~\ref{thr:18} in essentially the same step.\\
  Whereas in the proof of Lemma~\ref{le:1}, we explicitly use the fact
  that $H_\omega=H_0+V_\omega$ where $H_0$ is the free Laplace
  operator~\eqref{eq:29}, the proof we give of Lemma~\ref{le:9} still
  works if $H_0$ is any convolution matrix with exponentially decaying
  off diagonal coefficients if one replaces the condition $|E-E'|>2d$
  with the condition $|E-E'|>\sup\sigma(H_0)-\inf\sigma(H_0)$.
\end{Rem}
\section{Proof of the decorrelation estimates}
\label{sec:proof}
Before starting with the proofs of Lemma~\ref{le:1} and~\ref{le:9},
let us recall additional properties for the discrete Anderson model
known to be true under the assumptions we made on the distribution of
the random potential.
\subsection{Some facts on the discrete Anderson model}
\label{sec:some-facts-discrete}
Basic estimates on the distribution of the eigenvalues of the Anderson
model are the Wegner and Minami estimates.
\subsubsection{The Wegner and Minami estimates}
\label{sec:wegn-minami-estim}
One has
\begin{Th}[\cite{MR639135}]
  \label{thr:1}
  There exists $C>0$ such that, for $J\subset\R$, and $\Lambda$, a
  cube in $\Z^d$, one has
  \begin{equation}
    \label{eq:1}
    \E\left[\text{tr}(\car_J(H_\omega(\Lambda)))
    \right]\leq C |J|\,|\Lambda|
  \end{equation}
  where
  \begin{itemize}
  \item $H_\omega(\Lambda))$ is the operator $H_\omega$ restricted to
    $\Lambda$ with periodic boundary conditions,
  \item $\car_J(H)$ is the spectral projector of the operator $H$ on
    the energy interval $J$.
  \end{itemize}
\end{Th}
\noindent We refer to~\cite{MR2509108,MR2307751,MR2378428} for simple
proofs and more details on the Wegner estimate.\\
Another crucial estimate is the Minami estimate.
\begin{Th}[\cite{MR97d:82046,MR2360226,MR2290333,MR2505733}]
  \label{thr:2}
  There exists $C>0$ such that, for $J\subset K$, and $\Lambda$, a
  cube in $\Z^d$, one has
  \begin{equation}
    \label{eq:2}
    \E\left[\text{tr}(\car_J(H_\omega(\Lambda)))
      \cdot\text{tr}(\car_K(H_\omega(\Lambda))-1)\right]\leq
    C |J|\,|K|\,|\Lambda|^2.
  \end{equation}
\end{Th}
\noindent For $J=K$, the estimate~\eqref{eq:2} was proved
in~\cite{MR97d:82046,MR2360226,MR2290333,MR2505733}; for $J\not=K$, it
can be found in~\cite{MR2505733}. In their nature,~\eqref{eq:8}
or~\eqref{eq:16} and~\eqref{eq:2} are quite similar: the Minami
estimate can be interpreted as a decorrelation estimate for close
together eigenvalues. It can be used to obtain the counterparts of
Theorems~\ref{thr:4} and~\ref{thr:18} when $E$ and $E'$ tend to each
other as $|\Lambda|\to+\infty$ (see~\cite{Ge-Kl:10}).
\par Note that the Minami estimate~\eqref{eq:2} has been proved for
the discrete Anderson model on intervals $I$ irrelevant of the
spectral type of $H_\omega$ in $I$. Our proof of the decorrelation
estimates~\eqref{eq:8} and~\eqref{eq:16} makes use of the fact that
$I$ lies in the localized region.
\subsubsection{The localized regime}
\label{sec:localized-regime}
Let us now give a precise description of what we mean with the region
of localization or the localized regime. We prove
\begin{Pro}
  \label{pro:1}
  Recall that $I=I_+\cup I_-$ is the region of $\Sigma$ where the
  finite volume fractional moment criteria of~\cite{MR2002h:82051} for
  $H_\omega(\Lambda)$ are verified for $\Lambda$ sufficiently large.\\
  Then,
  \begin{description}
  \item[(Loc)] there exists $\nu>0$ such that, for any $p>0$, there
    exists $q>0$ and $L_0>0$ such that, for $L\geq L_0$, with
    probability larger than $1-L^{-p}$, if
    \begin{enumerate}
    \item $\varphi_{n,\omega}$ is a normalized eigenvector of
      $H_{\omega}(\Lambda_L)$ associated to an energy $E_{n,\omega}\in
      I$,
    \item $x_{n,\omega}\in \Lambda_L$ is a maximum of
      $x\mapsto|\varphi_{n,\omega}(x)|$ in $\Lambda_L$,
    \end{enumerate}
    then, for $x\in\Lambda_L$, one has
    \begin{equation}
      \label{eq:19}
      |\varphi_{n,\omega}(x)|\leq L^q e^{-\nu |x-x_{n,\omega}|}.
    \end{equation}
    The point $x_{n,\omega}$ is called a localization center for
    $\varphi_{n,\omega}$ or $E_{n,\omega}$.
  \end{description}
\end{Pro}
\noindent Note that, by Minami's estimate, the eigenvalues of
$H_\omega(\Lambda)$ are almost surely simple. Thus, we can associate a
localization center to an eigenvalue as it is done in
Proposition~\ref{pro:1}.\\
In its spirit, this result is not new (see
e.g.~\cite{MR2002h:82051,MR2203782,Ge-Kl:10}). We state it in a form
convenient for our purpose. We prove Proposition~\ref{pro:1} in
section~\ref{sec:proof-proposition}
\subsection{The proof of Lemmas~\ref{le:1} and~\ref{le:9}}
\label{sec:proof-theorems}
The basic idea of the proof is to show that, when $\omega$ varies, two
eigenvalues of $H_\omega(\Lambda)$ cannot vary in a synchronous
manner, or, put in another way, locally in $\omega$, if $E(\omega)$
and $E'(\omega)$ denote the two eigenvalues under consideration, for
some $\gamma$ and $\gamma'$, the mapping
$(\omega_\gamma,\omega_{\gamma'})\mapsto (E(\omega),E'(\omega))$ is a
local diffeomorphism when all the other random variables, that is
$(\omega_\alpha)_{\alpha\not\in\{\gamma,\gamma'\}}$, are fixed.\\
As we are in the localized regime, we will exploit this by noting that
eigenvalues of $H_\omega(\Lambda)$ can only depend significantly of
$(\log L)^{d}$ random variables i.e. we can study what happens in
cubes that are of side-length $\log L$ while the energy interval where
we want to control things are of size $L^{-d}$. This is the essence of
Lemma~\ref{le:3} below. This lemma is proved under the general
assumptions~\eqref{eq:1},~\eqref{eq:2} and (Loc). In particular, it is
valid for if one replaces the discrete Laplacian with any convolution
matrix with exponentially decaying off diagonal coefficients.\\
The second step consists in analyzing the mapping
$(\omega_\gamma,\omega_{\gamma'})\mapsto (E(\omega),E'(\omega))$ on
these smaller cubes. The main technical result is Lemma~\ref{le:5}
that shows that, under the conditions of Lemmas~\ref{le:1}
and~\ref{le:9}, with a large probability, eigenvalues away from each
other cannot move synchronously as functions of the random
variables. Of course, this will not be correct for all random models:
constructing artificial degeneracies, one can easily coin up random
models where this is not the case.\\
Lemmas~\ref{le:1} and~\ref{le:9} will be proved in essentially the
same way; the only difference will be in Lemma~\ref{le:5} that
controls the joint dependence of two distinct eigenvalues on the
random variables.
\par Let $J_L=E+L^{-d}[-1,1]$ and $J'_L=E'+L^{-d}[-1,1]$. Pick $L$
sufficiently large so that $J_L\subset I$ and $J'_L\subset I$ are
contained in $I$ where (Loc) holds true.\\
Pick $c L^\alpha\leq \ell\leq L^\alpha/c$ where $c>0$ is fixed. By
\eqref{eq:2}, we know that
\begin{equation*}
  \P\left(\#[\sigma(H_\omega(\Lambda_\ell))\cap
    J_L]\geq2 \text{ or }
    \#[\sigma(H_\omega(\Lambda_\ell))\cap
    J'_L ]\geq2\right)\leq C(\ell/L)^{2d}
\end{equation*}
where $\#[\cdot]$ denotes the cardinality of $\cdot$.\\
So if we define
\begin{equation*}
  \P_0=\P\left(\#[\sigma(H_\omega(\Lambda_\ell))\cap J_L]=1,
    \#[\sigma(H_\omega(\Lambda_\ell))\cap J'_L]=1\right),
\end{equation*}
it suffices to show that
\begin{equation}
  \label{eq:9}
  \P_0\leq C(\ell/L)^{2d}\cdot
  \begin{cases}
    e^{(\log L)^\beta}&\text{ if the dimension }d=1,\\
    (\log L)^C&\text{ if the dimension }d>1.
  \end{cases}
\end{equation}
First, using the assumption (Loc), we are going to reduce the proof
of~\eqref{eq:9} to the proof of a similar estimate where the cube
$\Lambda_\ell$ will be replaced by a much smaller cube, a cube of side
length of order $\log L$. We prove
\begin{Le}
  \label{le:3}
  There exists $C>0$ such that, for $L$ sufficiently large,
  \begin{equation*}
    \P_0 \leq C(\ell/L)^{2d}+C(\ell/\tilde\ell)^d\,\P_1
  \end{equation*}
  where $\tilde\ell=C\log L$ and
  \begin{equation*}
    \P_1:=\P(\#[\sigma(H_\omega(\Lambda_{\tilde\ell}))\cap
    \tilde J_L]\geq1)\text{ and
    }\#[\sigma(H_\omega(\Lambda_{\tilde\ell}))\cap
    \tilde J'_L]\geq1)
  \end{equation*}
  where $\tilde J_L=E+L^{-d}(-2,2)$ and $\tilde J'_L=E'+L^{-d}(-2,2)$.
\end{Le}
\begin{proof}[Proof of Lemma~\ref{le:3}]
  Fix $C>0$ large so that $e^{-C\gamma\log L/2}\leq L^{-2d-q}$ where
  $q$ and $\gamma$ are given by assumption (Loc) where we choose
  $p=d$. Let $\Omega_0$ be the set of probability $1-L^{-p}$ where (1)
  and (2) in assumption (Loc) are satisfied.  Define $\tilde\ell=C\log
  L$. We prove
  \begin{Le}
    \label{le:2}
    There exists a covering of $\Lambda_\ell$ by cubes, say
    $\Lambda_\ell=\cup_{\gamma\in\Gamma}[\gamma+\Lambda_{\tilde\ell}]$,
    such that $\#\Gamma\asymp(\ell/\tilde\ell)^d$, and, if
    $\omega\in\Omega_0$ is such that $H_\omega(\Lambda_\ell)$ has
    exactly one eigenvalue in $J_L$ and exactly one eigenvalue in
    $J'_L$, then
    \begin{enumerate}
    \item either, there exists $\gamma$ and $\gamma'$ such that
      $\gamma+\Lambda_{\tilde\ell}\cap\gamma'+\Lambda_{\tilde\ell}
      =\emptyset$ and
      \begin{itemize}
      \item $H_\omega(\gamma+\Lambda_{\tilde\ell})$ has exactly one
        e.v. in $\tilde J_L$
      \item $H_\omega(\gamma'+\Lambda_{\tilde\ell})$ has exactly one
        e.v. in $\tilde J'_L$.
      \end{itemize}
    \item or $H_\omega(\Lambda_{5\tilde\ell}(\gamma))$ has exactly one
      e.v. in $\tilde J_L$ and exactly one e.v. in $\tilde J'_L$.
    \end{enumerate}
  \end{Le}
  \noindent We postpone the proof of Lemma~\ref{le:2} to complete that
  of Lemma~\ref{le:3}. Using the estimate on $\P(\Omega_0)$, the
  independence of $H_\omega(\gamma+\Lambda_{\tilde\ell})$ and
  $H_\omega(\gamma'+\Lambda_{\tilde\ell})$ when alternative (1) is the
  case in Lemma~\ref{le:2}, Wegner's estimate~\eqref{eq:1} and the
  fact the random variables are identically distributed, we compute
  \begin{equation*}
    \begin{split}
      \P_0&\leq L^{-2d}+C(\ell/\tilde\ell)^d\P\left(
        \left\{ \begin{matrix}
            \sigma(H_\omega(\Lambda_{3\tilde\ell}(0)))\cap \tilde
            J_L\not=\emptyset\\
            \sigma(H_\omega(\Lambda_{3\tilde\ell}(0)))\cap \tilde
            J'_L\not=\emptyset
          \end{matrix}\right\}\right)\\
      &\hskip2cm+C(\ell/\tilde\ell)^{2d}
      \P(\#[\sigma(H_\omega(\Lambda_{\tilde\ell}(0)))\cap \tilde
      J_L]\geq1)\P(\#[\sigma(H_\omega(\Lambda_{\tilde\ell}(0)))\cap
      \tilde J'_L]\geq1)\\
      &\leq CL^{-2d}+C(\ell/\tilde\ell)^{2d}(\tilde\ell/L)^{2d}+
      C(\ell/\tilde\ell)^d\,\P_1\leq
      C(\ell/L)^{2d}+C(\ell/\tilde\ell)^d\,\P_1
    \end{split}
  \end{equation*}
  where $\P_1$ is defined in Lemma~\ref{le:3} for $5\tilde\ell$
  replaced with $\tilde\ell$. This completes the proof of
  Lemma~\ref{le:3}.
\end{proof}
\begin{proof}[Proof of Lemma~\ref{le:2}]
  For $\gamma\in\tilde\ell\Z^d\cap\Lambda_\ell$, consider the cubes
  $(\gamma+\Lambda_{\tilde\ell})_{\gamma\in\tilde\ell\Z^d
    \cap\Lambda_\ell}$. They cover $\Lambda_\ell$. Recall that we are
  taking periodic boundary conditions. If the localization centers
  associated to the two eigenvalues of $H_\omega(\Lambda_\ell)$
  assumed to be respectively in $\tilde J_L$ and $\tilde J'_L$ are at
  a distance less than $3\tilde\ell$ from one another, then we can
  find $\gamma\in\tilde\ell\Z^d$ such that both localization centers
  belong $\gamma+\Lambda_{4\tilde\ell}$ (for $\tilde\ell=C\log L$ and
  $C>0$ sufficiently large). Thus,
  by the localization property (Loc), we are in case (2). \\
  If the distance is larger than $3\tilde\ell$, we can find
  $\gamma\in\tilde\ell\Z^d$ and $\gamma'\in\tilde\ell\Z^d$ such that
  each of the cubes $\gamma+\Lambda_{\tilde\ell/2}$ and
  $\gamma'+\Lambda_{\tilde\ell/2}$ contains exactly one of the
  localization centers and
  $(\gamma+\Lambda_{\tilde\ell/2})\cap(\gamma'+
  \Lambda_{\tilde\ell/2})=\emptyset$. So for $\tilde\ell=C\,\log L$
  and $C>0$ sufficiently large, by
  the localization property (Loc), we are in case (1).\\
  This completes the proof of Lemma~\ref{le:2}.
\end{proof}
\noindent We now proceed with the proof of~\eqref{eq:9}. Therefore, by
Lemma~\ref{le:3}, it suffices to prove that $\P_1$, defined in
Lemma~\ref{le:3}, satisfies, for some $C>0$,
\begin{equation}
  \label{eq:30}
  \P_1\leq C(\tilde \ell/L)^{2d}\cdot
  \begin{cases}
    e^{\tilde\ell^\beta}&\text{ if the dimension }d=1,\\
    \tilde\ell^C&\text{ if the dimension }d>1.
  \end{cases}
\end{equation}
Let $(E_j(\omega,\tilde\ell))_{1\leq j\leq (2\tilde\ell+1)^d}$ be the
eigenvalues of $H_\omega(\Lambda_{\tilde\ell})$ ordered in an
increasing way and repeated according to multiplicity.\\
Assume that $\omega\mapsto E(\omega)$ is the only eigenvalue of
$H_\omega(\Lambda_{\tilde\ell})$ in $J_L$. In this case, by standard
perturbation theory arguments (see e.g.~\cite{MR1335452,MR58:12429c}),
we know that
\begin{enumerate}
\item $E(\omega)$ being simple, $\omega\mapsto E(\omega)$ is real
  analytic, and if
  $\omega\mapsto\varphi(\omega)=(\varphi(\omega;\gamma))
  _{\gamma\in\Lambda_{\tilde\ell}}$ denotes the associated normalized
  real eigenvector, it is also real analytic in $\omega$;
\item one has $\partial_{\omega_\gamma}E(\omega)=
  \varphi^2(\omega;\gamma)\geq0$ which, in particular, implies that
  \begin{equation}
    \label{eq:65}
    \|\nabla_\omega E(\omega)\|_{\ell^1}=1;
  \end{equation}
\item the Hessian of $E$ is given by Hess$_\omega
  E(\omega)=((h_{\gamma\beta}))_{\gamma,\beta}$ where
  \begin{itemize}
  \item $\displaystyle h_{\gamma,\beta}=-2$Re$\displaystyle\langle
    (H_\omega(\Lambda_{\tilde\ell})-E(\omega))^{-1}\psi_\gamma(\omega),
    \psi_\beta(\omega)\rangle$,
  \item $\psi_\gamma= \varphi(\omega;\gamma)\Pi(\omega)\delta_\gamma$
  \item $\Pi(\omega)$ is the orthogonal projector on the orthogonal to
    $\varphi(\omega)$.
  \end{itemize}
\end{enumerate}
We prove
\begin{Le}
  \label{le:4}
  There exists $C>0$ such that
  \begin{equation*}
    \|\text{Hess}_\omega(E(\omega))\|_{\ell^\infty\to\ell^1}\leq
    \frac{C}{\dist(E(\omega),\sigma(H_\omega(\Lambda_{\tilde\ell}))
      \setminus\{E(\omega)\})}.
  \end{equation*}
\end{Le}
\begin{proof}[Proof of Lemma~\ref{le:4}]
  First, note that, by definition, $H_\omega(\Lambda_{\tilde\ell})$
  depends on $(2\tilde\ell+1)^d$ random variables so that Hess$_\omega
  E(\omega)$ is a $(2\tilde\ell+1)^d\times (2\tilde\ell+1)^d$ matrix. Hence,
  for
  $a=(a_\gamma)_{\gamma\in\Lambda_{\tilde\ell}}\in\C^{\Lambda_{\tilde\ell}}$
  and
  $b=(b_\gamma)_{\gamma\in\Lambda_{\tilde\ell}}\in\C^{\Lambda_{\tilde\ell}}$,
  we compute
  \begin{equation*}
    \langle \text{Hess}_\omega E\,
    a,b\rangle=-2\langle(H_\omega(\Lambda_{\tilde\ell})-E(\omega))^{-1}
    \psi_a,\psi_b\rangle
  \end{equation*}
  where
  \begin{equation*}
    \psi_a= \Pi(\omega)
    \left(\sum_{\gamma\in\Lambda_{\tilde\ell}} a_\gamma
      |\delta_\gamma\rangle
      \langle\delta_\gamma|\right)\varphi(\omega)=
    \sum_{\gamma\in\Lambda_{\tilde\ell}} a_\gamma\varphi(\omega;\gamma)
    \Pi(\omega)\delta_\gamma.
  \end{equation*}
  Hence, $\|\psi_a\|_2\leq C\|a\|_\infty$ and, for some $C>0$,
  \begin{equation*}
    \|\text{Hess}_\omega(E(\omega))\|_{\ell^\infty\to\ell^1}\leq
    \frac{C}{\dist(E(\omega),\sigma(H_\omega(\Lambda_{\tilde\ell}))
      \setminus\{E(\omega)\})}.
  \end{equation*}
  This completes the proof of Lemma~\ref{le:4}.
\end{proof}
\noindent Note that, using~\eqref{eq:2}, Lemma~\ref{le:4} yields, for
$\varepsilon\in(4L^{-d},1)$,
\begin{equation*}
  \P\left(\left\{\omega;\
      \begin{matrix}
        \sigma(H_\omega(\Lambda_{\tilde\ell}))\cap
        \tilde J_L=\{E(\omega)\}\\
        \|\text{Hess}_\omega(E(\omega))\|_{\ell^\infty\to\ell^1}
        \geq\varepsilon^{-1}
      \end{matrix}
    \right\}\right)\leq C\varepsilon\tilde\ell^{2d} L^{-d}.
\end{equation*}
Hence, for $\varepsilon\in(4L^{-d},1)$, one has
\begin{equation}
  \label{eq:64}
  \P_1\leq C\varepsilon\tilde\ell^{2d} L^{-d}+\P_\varepsilon
\end{equation}
where
\begin{equation}
  \label{eq:31}
  \P_\varepsilon=\P(\Omega_0(\varepsilon))    
\end{equation}
and
\begin{equation}
  \label{eq:33}
  \Omega_0(\varepsilon)=\left\{\omega;\ 
    \begin{matrix}
      \sigma(H_\omega(\Lambda_{\tilde\ell}))\cap \tilde
      J_L=\{E(\omega)\}\\
      \{E(\omega)\}=\sigma(H_\omega(\Lambda_{\tilde\ell}))\cap
      (E-C\varepsilon,E+C\varepsilon),\\
      \sigma(H_\omega(\Lambda_{\tilde\ell}))\cap \tilde
      J'_L=\{E'(\omega)\}\\\{E'(\omega)\}
      =\sigma(H_\omega(\Lambda_{\tilde\ell}))\cap
      (E'-C\varepsilon,E'+C\varepsilon)
    \end{matrix}
  \right\}
\end{equation}
We will now estimate $\P_\varepsilon$. The basic idea is to prove that
the eigenvalues $E(\omega)$ and $E'(\omega)$ depend effectively on at
least two independent random variables. A simple way to guarantee this
is to ensure that their gradients with respect to $\omega$ are not
co-linear. In the present case, the gradients have non negative
components and their $\ell^1$-norm is $1$; hence, it suffices to prove
that they are different to ensure that they are not co-linear.\\
We prove
\begin{Le}
  \label{le:5}
  Let $L\geq1$. For the discrete Anderson model, one has
  \begin{enumerate}
  \item in any dimension $d$: for $\Delta E>2d$, if the random
    variables $(\omega_\gamma)_{\gamma\in\Lambda_L}$ are bounded by
    $K$, for $E_j(\omega)$ and $E_k(\omega)$, simple eigenvalues of
    $H_\omega(\Lambda_L)$ such that $|E_k(\omega)-E_j(\omega)|\geq
    \Delta E$, one has
    \begin{equation}
      \label{eq:37}
      \|\nabla_\omega(E_j(\omega)-E_k(\omega))\|_2\geq \frac{\Delta E-2d}{K}
      (2L+1)^{-d/2};
    \end{equation}
  \item in dimension 1: fix $E<E'$ and $\beta>1/2$; let $\P$ denote
    the probability that there exists $E_j(\omega)$ and $E_k(\omega)$,
    simple eigenvalues of $H_\omega(\Lambda_L)$ such that
    $|E_k(\omega)-E|+|E_j(\omega)-E'|\leq e^{-L^{\beta}}$ and such
    that
    \begin{equation}
      \label{eq:38}
      \|\nabla_\omega(E_j(\omega)-E_k(\omega))\|_1\leq e^{-L^{\beta}};
    \end{equation}
    then, there exists $c>0$ such that
    \begin{equation}
      \label{eq:39}
      \P\leq e^{-c L^{2\beta}}.
    \end{equation}
  \end{enumerate}
\end{Le}
\noindent We postpone the proof of Lemma~\ref{le:5} for a while to
estimate $\P_\varepsilon$. Set
\begin{equation}
  \label{eq:7}
  \lambda=\lambda_L=
  \begin{cases}
    e^{-\tilde\ell^{\beta}}&\text{ if the dimension }d=1,\\
    \frac{\Delta E-2d}{K}\tilde\ell^{-d/2}&\text{ if the dimension
    }d>1.
  \end{cases}
\end{equation}
For $\gamma$ and $\gamma'$ in $\Lambda_{\tilde\ell}$, define
\begin{equation}
  \label{eq:35}
  \Omega_{0,\beta}^{\gamma,\gamma'}(\varepsilon)=\Omega_0(\varepsilon)\cap
  \left\{\omega;\
    |J_{\gamma,\gamma'}(E(\omega),E'(\omega))
    |\geq \lambda\right\}
\end{equation}
where $J_{\gamma,\gamma'}(E(\omega),E'(\omega))$ is the Jacobian of
the mapping $(\omega_\gamma,\omega_{\gamma'})\mapsto
(E(\omega),E'(\omega))$ i.e.
\begin{equation*}
  J_{\gamma,\gamma'}(E(\omega),E'(\omega))=
  \left|\begin{matrix}
      \partial_{\omega_\gamma} E(\omega) & \partial_{\omega_{\gamma'}}
      E(\omega) \\ \partial_{\omega_\gamma} E'(\omega)
      &\partial_{\omega_{\gamma'}} E'(\omega)
    \end{matrix}\right|.
\end{equation*}
In section~\ref{sec:proof-lemma-1}, we prove
\begin{Le}
  \label{le:8}
  Pick $(u,v)\in(\R^+)^{2n}$ such that $\|u\|_1=\|v\|_1=1$. Then
  \begin{equation*}
    \max_{j\not=k}\left|\begin{matrix}u_j& u_k\\v_j&v_k\end{matrix}
    \right|^2\geq \frac1{4n^5}\|u-v\|^2_1.
  \end{equation*}
\end{Le}
\noindent We apply Lemma~\ref{le:5} with $L=\tilde\ell$ and
Lemma~\ref{le:8} to obtain that
\begin{equation}
  \label{eq:34}
  \P_\varepsilon\leq\sum_{\gamma\not=\gamma'}
  \P(\Omega_{0,\beta}^{\gamma,\gamma'}(\varepsilon))+\P_r
\end{equation}
where
\begin{enumerate}
\item in dimension $1$, we have $\displaystyle \P_r\leq
  C\tilde\ell^{2d}e^{-c\tilde\ell^{2\beta'}}$ for any
  $1/2<\beta'<\beta$; thus, for $L$ sufficiently large, as
  $\tilde\ell\geq c\log L$ and $\beta>1/2$, we have
  \begin{equation}
    \label{eq:63}
    \P_r\leq L^{-2d}.
  \end{equation}
\item in dimension $d$, as by assumption $\Delta E>2d$, one has
  $\P_r=0$, thus,~\eqref{eq:63} still holds.
\end{enumerate}
In the sequel, we will write
$\omega=(\omega_\gamma,\omega_{\gamma'},\omega_{\gamma,\gamma'})$
where
$\omega_{\gamma,\gamma'}=(\omega_\beta)_{\beta\not\in\{\gamma,\gamma'\}}$.\\
To estimate $\P(\Omega_{0,\beta}^{\gamma,\gamma'}(\varepsilon))$, we
use
\begin{Le}
  \label{le:6}
  Pick $\varepsilon=L^{-d}\lambda^{-3}$. For any
  $\omega_{\gamma,\gamma'}$, if there exists
  $(\omega^0_\gamma,\omega^0_{\gamma'})\in\R^2$ such that
  $(\omega^0_\gamma,\omega^0_{\gamma'},\omega_{\gamma,\gamma'})\in
  \Omega_{0,\beta}^{\gamma,\gamma'}(\varepsilon)$, then, for
  $(\omega_\gamma,\omega_{\gamma'})\in\R^2$ such that
  $|(\omega_\gamma,\omega_{\gamma'})-(\omega^0_\gamma,
  \omega^0_{\gamma'})|_\infty\geq L^{-d}\lambda^{-2}$, one has
  $(E_j(\omega),E_{j'}(\omega))\not\in \tilde J_L\times\tilde J'_L$.
\end{Le}
\noindent Recall that $g$ is the density of the random variables
$(\omega_\gamma)_\gamma$; it is assumed to be bounded and compactly
supported. Hence, the probability
$\P(\Omega_{0,\beta}^{\gamma,\gamma'}(\varepsilon))$ is estimated as
follows
\begin{equation}
  \label{eq:36}
  \begin{split}
    &\P(\Omega_{0,\beta}^{\gamma,\gamma'}(\varepsilon))
    =\E_{\gamma,\gamma'}\left(\int_{\R^2}
      \car_{\Omega_{0,\beta}^{\gamma,\gamma'}(\varepsilon)}(\omega)
      g(\omega_\gamma) g(\omega_{\gamma'})d\omega_\gamma
      d\omega_{\gamma'}\right)\\
    &\leq\E_{\gamma,\gamma'}\left(\int_{|(\omega_\gamma,
        \omega_{\gamma'})- (\omega^0_\gamma,
        \omega^0_{\gamma'})|_\infty<L^{-d}\lambda^{-2}}
      g(\omega_\gamma) g(\omega_{\gamma'})d\omega_\gamma
      d\omega_{\gamma'}\right)\\
    &\leq C L^{-2d}\lambda^{-4}
  \end{split}
\end{equation}
where $\E_{\gamma,\gamma'}$ denotes the expectation with respect to
all the random variables except $\omega_\gamma$ and
$\omega_{\gamma'}$.\\
Summing~\eqref{eq:36} over
$(\gamma,\gamma')\in\Lambda_{\tilde\ell}^2$, using~\eqref{eq:34}
and~\eqref{eq:63}, we obtain
\begin{equation*}
  \P_\varepsilon\leq C L^{-2d}\lambda^{-4}.
\end{equation*}
We now plug this into~\eqref{eq:64} and use the fact that
$\varepsilon=L^{-d}\lambda^{-3}$ to complete the proof
of~\eqref{eq:30}. This completes the proofs of Lemmas~\ref{le:1}
and~\ref{le:9}.\qed
\begin{proof}[Proof of Lemma~\ref{le:6}]
  Recall that, for any $\gamma$, $\omega_\gamma\mapsto E_j(\omega)$
  and $\omega_\gamma\mapsto E_{j'}(\omega)$ are non decreasing. Hence,
  to prove Lemma~\ref{le:6}, it suffices to prove that, for
  $|(\omega_\gamma,\omega_{\gamma'})-(\omega^0_\gamma,
  \omega^0_{\gamma'})|_\infty=L^{-d}\lambda^{-2}$, one has
  $(E_j(\omega),E_{j'}(\omega))\not\in \tilde J_L\times\tilde J'_L$.\\
  Let $\mathcal{S}_\beta$ denote the square
  $\mathcal{S}_\beta=\{|(\omega_\gamma,\omega_{\gamma'})-(\omega^0_\gamma,
  \omega^0_{\gamma'})|_\infty\leq L^{-d}\lambda^{-2}\}$.\\
  Recall that $\varepsilon=L^{-d}\lambda^{-3}$. Pick
  $\omega_{\gamma,\gamma'}$ such that there exists
  $(\omega^0_\gamma,\omega^0_{\gamma'})\in\R^2$ for which one has
  $(\omega^0_\gamma,\omega^0_{\gamma'},\omega_{\gamma,\gamma'})\in
  \Omega_{0,\beta}^{\gamma,\gamma'}(\varepsilon)$. To shorten the
  notations, in the sequel, we write only the variables
  $(\omega_\gamma,\omega_{\gamma'})$ as $\omega_{\gamma,\gamma'}$
  stays fixed throughout the proof; e.g. we write $E((\omega_\gamma,
  \omega_{\gamma'}))$ instead of $E((\omega_\gamma,
  \omega_{\gamma'},\omega_{\gamma,\gamma'}))$.\\
  Consider the mapping $(\omega_\gamma,\omega_{\gamma'})\mapsto
  \varphi(\omega_\gamma,\omega_{\gamma'}):=(E(\omega),E'(\omega))$. We
  will show that $\varphi$ defines an analytic diffeomorphism form
  $\mathcal{S}_\beta$
  to $\varphi(\mathcal{S}_\beta)$.\\
  By~\eqref{eq:35} and~\eqref{eq:33}, the definitions of
  $\Omega_{0,\beta}^{\gamma,\gamma'}(\varepsilon)$ and
  $\Omega_0(\varepsilon)$, we know that
  \begin{equation*}
    \begin{split}
      \sigma(H_{(\omega^0_\gamma,\omega^0_{\gamma'})}
      (\Lambda_{\tilde\ell}))&\cap (E-C\varepsilon,E+C\varepsilon)
      =\{E(\omega)\}\subset (E-CL^{-d},E+CL^{-d}) ,\\
      \sigma(H_{(\omega^0_\gamma,\omega^0_{\gamma'})}
      (\Lambda_{\tilde\ell}))&\cap
      [(E-C\varepsilon,E-C\varepsilon/2)\cup
      (E+C\varepsilon/2,E+C\varepsilon)]=\emptyset,\\
      \sigma(H_{(\omega^0_\gamma,\omega^0_{\gamma'})}
      (\Lambda_{\tilde\ell}))&\cap (E'-C\varepsilon,E'+C\varepsilon)
      =\{E'(\omega)\}\subset (E'-CL^{-d},E'+CL^{-d}) ,\\
      \sigma(H_{(\omega^0_\gamma,\omega^0_{\gamma'})}
      (\Lambda_{\tilde\ell}))&\cap
      [(E'-C\varepsilon,E'-C\varepsilon/2)\cup
      (E'+C\varepsilon/2,E'+C\varepsilon)]=\emptyset.
    \end{split}
  \end{equation*}
  By~\eqref{eq:65}, as $L^{-d}\lambda^{-2}\leq \lambda\varepsilon$,
  for $(\omega_\gamma,\omega_{\gamma'})\in\mathcal{S}_\beta$ , one has
  \begin{equation*}
    \begin{split}
      \sigma(H_\omega(\Lambda_{\tilde\ell}))&\cap
      (E-C\varepsilon/2,E+C\varepsilon/2)=\{E(\omega)\} \subset
      (E-C\varepsilon/4,E+C\varepsilon/4),\\
      \sigma(H_\omega(\Lambda_{\tilde\ell}))&\cap
      [(E-C\varepsilon/2,E-C\varepsilon/4)\cup
      (E+C\varepsilon/4,E+C\varepsilon/2)]=\emptyset,\\
      \sigma(H_\omega(\Lambda_{\tilde\ell}))&\cap
      (E'-C\varepsilon/2,E'+C\varepsilon/2)=\{E(\omega)\}\subset
      (E'-C\varepsilon/4,E'+C\varepsilon/4),\\
      \sigma(H_\omega(\Lambda_{\tilde\ell}))&\cap
      [(E'-C\varepsilon/2,E'-C\varepsilon/4)\cup
      (E'+C\varepsilon/4,E'+C\varepsilon/2)]=\emptyset.
    \end{split}
  \end{equation*}
  Hence, by Lemma~\ref{le:4}, for
  $(\omega_\gamma,\omega_{\gamma'})\in\mathcal{S}_\beta$, one has
  \begin{equation*}
    \|\text{Hess}_\omega(E(\omega))\|_{\ell^\infty\to\ell^1}+
    \|\text{Hess}_\omega(E'(\omega))\|_{\ell^\infty\to\ell^1}
    \leq C\varepsilon^{-1}\leq CL^d\lambda^3.
  \end{equation*}
  By~\eqref{eq:65} and the Fundamental Theorem of Calculus, for
  $(\omega_\gamma,\omega_{\gamma'})\in\mathcal{S}_\beta$, we get that,
  \begin{equation}
    \label{eq:66}
    \begin{split}
      &\|\nabla\varphi(\omega_\gamma,
      \omega_{\gamma'})-\nabla\varphi(\omega^0_\gamma,
      \omega^0_{\gamma'})\|\\ &\leq
      \left(\|\text{Hess}_\omega(E(\omega))\|_{\ell^\infty\to\ell^1}+
        \|\text{Hess}_\omega(E'(\omega))\|_{\ell^\infty\to\ell^1}\right)
      L^{-d}\lambda^{-1} \leq C\lambda^2.
    \end{split}
  \end{equation}
  Let us show that $\varphi$ is one-to-one on the square
  $\mathcal{S}_\beta$. Using~\eqref{eq:66}, we compute
  \begin{equation*}
    \left|\varphi(\omega'_\gamma,\omega'_{\gamma'})-
      \varphi(\omega_\gamma,\omega_{\gamma'})-
      \nabla\varphi(\omega^0_\gamma,
      \omega^0_{\gamma'})\cdot\left(
        \begin{matrix}
          \omega'_\gamma-\omega_\gamma\\
          \omega'_{\gamma'}-\omega_{\gamma'}
        \end{matrix}\right)\right|\leq
    \lambda^2\left\|\begin{pmatrix}
        \omega'_\gamma-\omega_\gamma\\
        \omega'_{\gamma'}-\omega_{\gamma'}        
      \end{pmatrix}\right\|
  \end{equation*}
  As $(\omega^0_\gamma,\omega^0_{\gamma'},\omega_{\gamma,\gamma'})\in
  \Omega_{0,\beta}^{\gamma,\gamma'}(\varepsilon)$, we have
  \begin{equation*}
    \left|\text{Jac}\,\varphi(\omega^0_\gamma,
      \omega^0_{\gamma'})\right|\geq \lambda.
  \end{equation*}
  Hence, for $\tilde\ell$ large, we have
  \begin{equation*}
    \left|\varphi(\omega'_\gamma,\omega'_{\gamma'})-
      \varphi(\omega_\gamma,\omega_{\gamma'})\right|\geq
    \frac12\lambda\left\|\begin{pmatrix} 
        \omega'_\gamma-\omega_\gamma\\
        \omega'_{\gamma'}-\omega_{\gamma'}        
      \end{pmatrix}\right\|
  \end{equation*}
  so $\varphi$ is one-to-one. The estimate~\eqref{eq:66} yields
  \begin{equation*}
    |\text{Jac}\,\varphi(\omega_\gamma,
    \omega_{\gamma'})-\text{Jac}\,\varphi(\omega^0_\gamma,
    \omega^0_{\gamma'})|\leq \lambda^2
  \end{equation*}
  As $(\omega^0_\gamma,\omega^0_{\gamma'},\omega_{\gamma,\gamma'})\in
  \Omega_{0,\beta}^{\gamma,\gamma'}(\varepsilon)$, for $L$
  sufficiently large, this implies that
  \begin{equation}
    \label{eq:68}
    \forall(\omega_\gamma,\omega_{\gamma'})\in\mathcal{S}_\beta,
    \quad|J_{\gamma,\gamma'}(E(\omega),E'(\omega)) 
    |\geq\frac12\lambda.
  \end{equation}
  The Local Inversion Theorem then guarantees that $\varphi$ is an
  analytic diffeomorphism from $\mathcal{S}_\beta$ onto
  $\varphi(\mathcal{S}_\beta)$. By~\eqref{eq:68}, the Jacobian matrix
  of its inverse is bounded by $C\tilde\ell^{\beta}$ for some $C>0$
  independent of $L$. Hence, if for some
  $|(\omega_\gamma,\omega_{\gamma'})-(\omega^0_\gamma,
  \omega^0_{\gamma'})|_\infty= L^{-d}\lambda^{-2}$, one has
  $(E(\omega),E'(\omega))\in \tilde J_L\times\tilde J'_L$, then
  \begin{equation*}
      L^{-d}\lambda^{-2}=|(\omega_\gamma,\omega_{\gamma'})
      -(\omega^0_\gamma,\omega^0_{\gamma'})|_\infty
      =|\varphi^{-1}(E(\omega),E'(\omega))-\varphi^{-1}(E,E')|_\infty
      \leq CL^{-d}\lambda^{-1}
  \end{equation*}
  which is absurd when $L\to+\infty$ as $\lambda=\lambda_L\to0$
  (see~\eqref{eq:7}). This completes the proof of Lemma~\ref{le:6}.
\end{proof}
\subsection{Proof of Lemma~\ref{le:5}}
\label{sec:proof-lemma}
A fundamental difference between the points (1) and (2) in
Lemma~\ref{le:5} is that to prove point (2), we will the fact that
$H_0$ is the discrete Laplacian. In the proof of point (1), we can
take $H_0$ to be any convolution matrix with exponentially decaying
off diagonal coefficients if one replaces the condition $|E-E'|>2d$
with the condition $|E-E'|>\sup\sigma(H_0)-\inf\sigma(H_0)$.\\
As it is simpler, we start with the proof of point (1).
\subsubsection{The proof of point (1)}
\label{sec:proof-point-1}
Let $E_j(\omega)$ and $E_k(\omega)$ be simple eigenvalues of
$H_\omega(\Lambda_L)$ such that $|E_k(\omega)-E_j(\omega)|\geq \Delta
E>2d$. Then, $\omega\mapsto E_j(\omega)$ and $\omega\mapsto
E_k(\omega)$ are real analytic functions. Let $\omega\mapsto
\varphi_j(\omega)$ and $\omega\mapsto \varphi_k(\omega)$ be normalized
eigenvectors associated respectively to $E_j(\omega)$ and
$E_k(\omega)$. Differentiating the eigenvalue equation in $\omega$,
one computes
\begin{equation*}
  \begin{split}
    \omega\cdot\nabla_\omega(E_j(\omega)&-E_k(\omega)) =\langle
    V_\omega\varphi_j(\omega),\varphi_j(\omega)\rangle -\langle
    V_\omega\varphi_k(\omega),\varphi_k(\omega)\rangle\\&=E_j(\omega)-E_k(\omega)+
    \langle -\Delta \varphi_k(\omega),\varphi_k(\omega)\rangle-\langle
    -\Delta \varphi_j(\omega),\varphi_j(\omega)\rangle.
  \end{split}
\end{equation*}
As $0\leq -\Delta\leq 2d$ and as $\varphi_j(\omega)$ and
$\varphi_k(\omega)$ are normalized, we get that
\begin{equation*}
  \Delta E-2d\leq |E_j(\omega)-E_k(\omega)| -2d \leq 
  |\omega\cdot\nabla_\omega(E_j(\omega)-E_k(\omega))|.
\end{equation*}
Hence, as the random variables $(\omega_\gamma)_{\gamma\in\Lambda}$
are bounded, the Cauchy Schwartz inequality yields
\begin{equation*}
  \|\nabla_\omega(E_j(\omega)-E_k(\omega))\|_2\geq \frac{\Delta E-2d}{K}
  (2L+1)^{-d/2}.
\end{equation*}
which completes the proof of~\eqref{eq:37}.
\subsubsection{The proof of point (2)}
\label{sec:proof-point-3}
Let us now assume $d=1$. Fix $E<E'$. Pick $E_j(\omega)$ and
$E_k(\omega)$, simple eigenvalues of $H_\omega(\Lambda_L)$ such that
$|E_k(\omega)-E|+|E_j(\omega)-E'|\leq e^{-L^{\beta}}$. Then,
$\omega\mapsto E_j(\omega)$ and $\omega\mapsto E_k(\omega)$ are real
analytic functions. Let $\omega\mapsto \varphi^j(\omega)$ and
$\omega\mapsto \varphi^k(\omega)$ be normalized eigenvectors
associated respectively to $E_j(\omega)$ and $E_k(\omega)$. One
computes
\begin{equation*}
  \nabla_\omega E_j(\omega)=([\varphi^j(\omega;\gamma)]^2)_{\gamma\in\Lambda_L}
  \quad\text{and}\quad
  \nabla_\omega E_k(\omega)=([\varphi^k(\omega;\gamma)]^2)_{\gamma\in\Lambda_L}.
\end{equation*}
Hence, if
\begin{equation}
  \label{eq:60}
    e^{-L^{\beta}}\geq \|\nabla_\omega(E_j(\omega)-E_k(\omega))\|_1
    =\sum_{\gamma\in\Lambda_L}|\varphi^j(\omega;\gamma)-
    \varphi^k(\omega;\gamma)|
    \cdot|\varphi^j(\omega;\gamma)+\varphi^k(\omega;\gamma)|
\end{equation}
as $\|\nabla_\omega E_j(\omega)\|=\|\nabla_\omega E_k(\omega)\|=1$,
there exists a partition of $\Lambda_L=\{-L,\cdots,L\}$, say
$\mathcal{P}\subset\Lambda_L$ and $\mathcal{Q}\subset\Lambda_L$ such
that $\mathcal{P}\cup\mathcal{Q}=\Lambda_L$ and
$\mathcal{P}\cap\mathcal{Q}=\emptyset$ and such that
\begin{itemize}
\item for $\gamma\in\mathcal{P}$,
  $|\varphi^j(\omega;\gamma)-\varphi^k(\omega;\gamma)|\leq
  e^{-L^{\beta}/2}$;
\item for $\gamma\in\mathcal{Q}$,
  $|\varphi^j(\omega;\gamma)+\varphi^k(\omega;\gamma)|\leq
  e^{-L^{\beta}/2}$.
\end{itemize}
Introduce the orthogonal projectors $P$ and $Q$ defined by
\begin{equation*}
  P=\sum_{\gamma\in\mathcal{P}}|\gamma\rangle\langle\gamma|
  \quad\text{ and }\quad
  Q=\sum_{\gamma\in\mathcal{Q}}|\gamma\rangle\langle\gamma|.
\end{equation*}
One has
\begin{equation*}
  \|P\varphi^j-P\varphi^k\|_{2}\leq \sqrt{L}\,e^{-L^{\beta}/2}\quad\text{ and
  }\quad \|Q\varphi^j+Q\varphi^k\|_{2}\leq \sqrt{L}\,e^{-L^{\beta}/2}.
\end{equation*}
Clearly $\|P\varphi^j\|^2+\|Q\varphi^j\|^2=\|\varphi^j\|^2=1$. As
$\langle\varphi^j,\varphi^k\rangle=0$, one has
\begin{equation*}
  \begin{split}
    0&=\langle(P+Q)\varphi^j,(P+Q)\varphi^k\rangle= \langle
    P\varphi^j,P\varphi^k\rangle+\langle
    Q\varphi^j,Q\varphi^k\rangle\\
    &=\|P\varphi^j\|^2-\|Q\varphi^j\|^2
    +O\left(\sqrt{L}\,e^{-L^{\beta}/2}\right).
  \end{split}
\end{equation*}
Hence
\begin{equation*}
  \|P\varphi^j\|^2=\frac12+O(\sqrt{L}\,e^{-L^{\beta}/2})\text{ and
  }\|Q\varphi^j\|^2=\frac12+O(\sqrt{L}\,e^{-L^{\beta}/2}).
\end{equation*}
This implies that
\begin{equation}
  \label{eq:52}
  \mathcal{P}\not=\emptyset\text{ and }\mathcal{Q}\not=\emptyset.
\end{equation}
We set $h_-=P\varphi^j-P\varphi^k$ and
$h_+=Q\varphi^j+Q\varphi^k$. The eigenvalue equations for
$E_j(\omega)$ and $E_k(\omega)$ yields
\begin{equation*}
  (-\Delta+W_\omega)\varphi^j=\Delta
  E(\omega)\varphi^j\text{ and }(-\Delta+W_\omega)\varphi^k=-\Delta
  E(\omega)\varphi^k
\end{equation*}
where
\begin{equation*}
  \Delta E(\omega)=(E_j(\omega)-E_k(\omega))/2,\quad
  W_\omega=V_\omega-\overline{E}(\omega),\quad
  \overline{E}(\omega)=(E_j(\omega)+E_k(\omega))/2.  
\end{equation*}
To simplify the notation, from now on, we write $u=\varphi^j$; then,
one has $\varphi^k=Pu-Qu+O(\sqrt{L}\,e^{-L^{\beta}/2})$. This yields
\begin{equation*}
  \begin{cases}
    (-\Delta+W_\omega)(Pu+Qu)&=\Delta
    E(\omega)(Pu+Qu),\\
    (-\Delta+W_\omega)(Pu-Qu+h_--h_+)&=-\Delta
    E(\omega)(Pu-Qu+h_--h_+)
  \end{cases}
\end{equation*}
that is
\begin{equation*}
  \begin{cases}
    (-\Delta+W_\omega)(Pu)&=\Delta E(\omega)Qu-h,\\
    (-\Delta+W_\omega)(Qu)&=\Delta E(\omega)Pu+h
  \end{cases}
\end{equation*}
where $h:=(-\Delta+W_\omega-\Delta E(\omega))(h_--h_+)/2$.  As $P
W_\omega Q=0$, this can also be written as
\begin{equation}
  \label{eq:32}
  \begin{cases}
    [-(P\Delta Q+Q\Delta P)-\Delta E]u&=h_1,\\
    [-(P\Delta P+Q\Delta Q)+V_\omega-\overline{E}]u&=h_2.
  \end{cases}
\end{equation}
where
\begin{gather*}
  h_1:=(P-Q)h+(\Delta E(\omega)-\Delta E)u,\quad h_2:=
  (Q-P)h+(\overline{E}(\omega)-\overline{E})u,\\
  \Delta E=(E'-E)/2,\quad\quad\overline{E}=(E+E')/2.
\end{gather*}
By our assumption on $E_j(\omega)$ and $E_k(\omega)$, we know that
\begin{equation*}
  |\Delta E(\omega)-\Delta E|\leq 2e^{-L^{\beta}},\quad
  |\overline{E}(\omega)-\overline{E}|\leq e^{-L^{\beta}},
  \quad\|h\|\leq C\sqrt{L}\,e^{-L^{\beta}/2}.
\end{equation*}
Hence, we get that
\begin{equation}
  \label{eq:50}
  \|h_1\|+\|h_2\|\leq C \sqrt{L}\,e^{-L^{\beta}/2}.
\end{equation}
So the above equations imply that
\begin{itemize}
\item $\Delta E$ is at a distance at most $\sqrt{L}\,e^{-L^{\beta}/2}$
  to the spectrum of the deterministic operator $-(P\Delta Q+Q\Delta
  P)$,
\item $u$ is close to being in the eigenspace associated to the
  eigenvalues close to $\Delta E$,
\item finally, $u$ is close to being in the kernel of the random
  operator $-(P\Delta P+Q\Delta Q)+V_\omega-\overline{E}$.
\end{itemize}
The firsts conditions will be used to describe $u$. The last condition
will be interpreted as a condition determining the random variables
$\omega_\gamma$ for sites $\gamma$ such that $|u_\gamma|$ is not too
small. We will show that the number of these sites is of size the
volume of the cube $\Lambda_L$; so, the probability that the second
equation in~\eqref{eq:32} be satisfied
should be very small.\\
To proceed, we first study the operator $-P\Delta Q-Q\Delta P$. As we
consider periodic boundary conditions, we compute
\begin{equation}
  \label{eq:43}
    -P\Delta Q-Q\Delta P =\sum_{\gamma\in\partial\mathcal{P}}
    (|\gamma+1\rangle\langle\gamma|+|\gamma\rangle\langle\gamma+1|)
    +\sum_{\gamma\in\partial\mathcal{Q}}
    (|\gamma+1\rangle\langle\gamma|+|\gamma\rangle\langle\gamma+1|)
\end{equation}
where $\partial\mathcal{P}=\{\gamma\in\mathcal{P};\
\gamma+1\in\mathcal{Q}\}\subset\mathcal{P}$ and
$\partial\mathcal{Q}=\{\gamma\in\mathcal{Q};\
\gamma+1\in\mathcal{P}\}\subset\mathcal{Q}$. By~\eqref{eq:52}, we know
that $\partial\mathcal{P}\not=\emptyset$ and
$\partial\mathcal{Q}\not=\emptyset$. \\
We first note that
$\partial\mathcal{P}\cap\partial\mathcal{Q}=\emptyset$. Here, as we
are considering the operators with periodic boundary conditions on
$\Lambda_L$, we identify $\Lambda_L$ with $\Z/L\Z$.\\
For $\mathcal{A}\subset \Lambda_L$ we define $\mathcal{A}+1=\{p+1;\
p\in\mathcal{A}\}$ to be the shift by one of $\mathcal{A}$. By
definition, $(\partial\mathcal{P}+1)\subset\mathcal{Q}$ and
$(\partial\mathcal{Q}+1)\subset\mathcal{P}$. Hence,
$(\partial\mathcal{P}+1)\cap\partial\mathcal{P}=\emptyset$ and
$(\partial\mathcal{Q}+1)\cap\partial\mathcal{Q}=\emptyset$.\\
Consider the set
$\mathcal{C}:=\partial\mathcal{P}\cup\partial\mathcal{Q}$. We can
partition it into its ``connected components'' i.e. $\mathcal{C}$ can
be written as a disjoint union of intervals of integers, say
$\displaystyle\mathcal{C}=\cup_{l=1}^{l_0} \mathcal{C}^c_l$. Then, by
the definition of $\partial\mathcal{P}$ and $\partial\mathcal{Q}$, for
$l\not=l'$, one has,
\begin{equation}
  \label{eq:41}
  \mathcal{C}^c_l\cap\mathcal{C}^c_{l'}=
  \mathcal{C}^c_l\cap(\mathcal{C}^c_{l'}+1)=\emptyset.  
\end{equation}
Define
$\mathcal{C}_l=\mathcal{C}^c_l\cup(\mathcal{C}^c_l+1)$.~\eqref{eq:41}
implies that, for $l\not=l'$,
\begin{equation}
  \label{eq:42}
  \mathcal{C}_l\cap\mathcal{C}_{l'}=\emptyset .
\end{equation}
Note that one may have
$\displaystyle\cup_{l=1}^{l_0}\mathcal{C}_l=\Lambda_L$. The
representation~\eqref{eq:43} then implies that the following block
decomposition
\begin{equation}
  \label{eq:40}
  -P\Delta Q-Q\Delta P=-\sum_{l=1}^{l_0} C_l\Delta C_l
\end{equation}
where $C_l$ is the projector $\displaystyle
C_l=\sum_{\gamma\in\mathcal{C}_j}|\gamma\rangle\langle\gamma|$. \\
Note that, by~\eqref{eq:42}, the projectors $C_l$ and $C_{l'}$ are
orthogonal to each other for $l\not=l'$. So the spectrum of the
operator $-P\Delta Q-Q\Delta P$ is given by the union of the spectra
of $(C_l\Delta C_l)_{1\leq l\leq l_0}$. Each of these operators is the
Dirichlet Laplacian on an interval of length $\#\mathcal{C}_l$. Its
spectral decomposition can be computed explicitly. We will use some
facts from this decomposition that we state now.
\begin{Le}
  \label{le:10}
  On a segment of length $n$, the Dirichlet Laplacian $\Delta_n$
  i.e. the $n\times n$ matrix
  \begin{equation*}
    \Delta_n=\begin{pmatrix}
      0& 1 &0 &\cdots &\cdots &0\\1&0& 1 &0 && \\ 0& 1 &0 &\ddots&\ddots
      &\vdots\\ \vdots& 0 &\ddots &\ddots &1&0\\ \vdots&&\ddots &1&0&1\\ 0&
      \cdots &\cdots&0 &1 &0
    \end{pmatrix}
  \end{equation*}
  satisfies
  \begin{itemize}
  \item its eigenvalues are simple and are given by
    $(2\cos(k\pi/(n+1)))_{1\leq k\leq n}$;
  \item for $k\in\{1,\cdots,n\}$, the eigenspace associated to
    $2\cos(k\pi/(n+1))$ is generated by the vector $(\sin[k
    j\pi/(n+1)])_{1\leq j\leq n}$.
  \end{itemize}
  Moreover, there exists $K_1>0$ such that, for any $n\geq1$, one has
  \begin{equation}
    \label{eq:17}
    \inf_{1\leq k<k'\leq n}
    \left|2\cos\left(\frac{k\pi}{n+1}\right)-
      2\cos\left(\frac{k'\pi}{n+1}\right)\right|\geq\frac1{K_1n^2}.
  \end{equation}
\end{Le}
\begin{proof}[Proof of Lemma~\ref{le:10}] The first statement follows
  immediately from the identity
  \begin{equation*}
    \sin\left(\frac{k(j+1)\pi}{n+1}\right)+
    \sin\left(\frac{k(j-1)\pi}{n+1}\right)
    =2\cos\left(\frac{k\pi}{n+1}\right)\sin\left(\frac{kj\pi}{n+1}\right).
  \end{equation*}
  The estimate~\eqref{eq:17} is an immediate consequence of
  \begin{equation*}
    \cos\left(\frac{k\pi}{n+1}\right)-
    \cos\left(\frac{k'\pi}{n+1}\right)=-
    2\sin\left(\frac{(k+k')\pi}{2(n+1)}\right)
    \sin\left(\frac{(k-k')\pi}{2(n+1)}\right).
  \end{equation*}
\end{proof}
\noindent We now solve the first equation in~\eqref{eq:32} that is
describe $u$ solution to this equation.
\begin{Le}
  \label{le:11}
  Let $u$ be a solution to~\eqref{eq:32} such that $\|u\|=1$. Then,
  for $L$ sufficiently large, one has
  \begin{equation}
    \label{eq:14}
    \left\|u-\sum_{l=1}^{l_0}C_lu\right\|\leq e^{-L^{\beta}/3}
  \end{equation}
  where, if for $1\leq l\leq l_0$, we write
  $\mathcal{C}_l=\{\gamma_l^-,\cdots,\gamma_l^+\}$
  ($n_l=\gamma_l^+-\gamma_l^-+1$), then,
  \begin{itemize}
  \item either there exists a unique $k_l\in\{1,\cdots,n_l\}$
    satisfying
    \begin{equation}
      \label{eq:5}
      \left|2\cos\left(\frac{k_l\pi}{n_l+1}\right)-\Delta E\right|
      <\frac1{K_1n^2}
    \end{equation}
    and $\alpha^l\in\R$ such that
    \begin{equation}
      \label{eq:45}
      \|C_lu-\alpha^lu^l\|\leq e^{-L^{\beta}/3}
    \end{equation}
    where
    \begin{equation*}
      u^l_\gamma=
      \begin{cases}
        \sin\left(\frac{k_l(\gamma-\gamma^-_l+1)\pi}
          {n_l+1}\right)&\text{ if }\gamma\in C_l,\\
        0&\text{ if }\gamma\not\in C_l.
      \end{cases}
    \end{equation*}
  \item there exists no $k_l\in\{1,\cdots,n_l\}$
    satisfying~\eqref{eq:5} then
    \begin{equation*}
      \|C_lu\|\leq e^{-L^{\beta}/3}.
    \end{equation*}
  \end{itemize}
\end{Le}
\begin{proof}[Proof of Lemma~\ref{le:11}]
  By Lemma~\ref{le:10}, the spacing between consecutive eigenvalues of
  $-C_l\Delta C_l$ is bounded below by $1/(K_1n^2)$.\\
  Let $\D C^\perp=1-\sum_{l=1}^{l_0}C_l$. Hence, $\D
  u=\sum_{l=1}^{l_0}C_l u+C^\perp u$, the terms in this sums being two
  by two orthogonal to each other. As $\Delta E>0$, the first equation
  in~\eqref{eq:32} then yields
  \begin{equation}
    \label{eq:21}
    \forall1\leq l\leq l_0,\quad \|-C_l\Delta C_l u-\Delta E\,
    C_lu\| \leq C\sqrt{L}\,e^{-L^{\beta}/2}\quad\text{and}\quad
    \|C^\perp u\|\leq C\sqrt{L}\,e^{-L^{\beta}/2}.
  \end{equation}
  Write $\mathcal{C}_l=\{\gamma_l^-,\gamma_l^-+1,\cdots,\gamma_l^+\}$ where one may
  have $\gamma_l^-=\gamma_l^+$. We assume that the $(\mathcal{C}_l)_{1\leq l\leq
    l_0}$ are ordered so that $\gamma_l^+<\gamma_{l+1}^-$.\\
  By the characterization of the spectrum of $-C_l\Delta C_l$,
  \begin{itemize}
  \item if $2\cos(k_l\pi/(n+1))$ is an eigenvalue of $-C_l\Delta C_l$
    closer to $\Delta E$ than a distance $L^{-2}/4K_1$ (by the remark
    made above, such an eigenvalue is unique), then, for some
    $\alpha^l$ real, one has
    \begin{equation*}
      \|C_lu-\alpha^lu^l\|\leq CL^{5/2}\,e^{-L^{\beta}/2}.
    \end{equation*}
  \item if there is no such eigenvalue, then
    \begin{equation}
      \label{eq:12}
      \|C_l u\|\leq C L^{5/2}\,e^{-L^{\beta}/2}.
    \end{equation}
  \end{itemize}
  This completes the proof of Lemma~\ref{le:11}.
\end{proof}
\noindent We now prove that $|u_\gamma|$ cannot be really small for
too many $\gamma$.
\begin{Le}
  \label{le:12}
  There exists $c>0$ such that, for $L$ sufficiently large,
  \begin{enumerate}
  \item either $\D\#\mathcal{C}\geq L/3$ and, for
    $\gamma\in\mathcal{C}$, $|u_\gamma|\geq e^{-L^{\beta}/6}$,
  \item or $l_0\geq 2cL^\beta$ and there exists
    $l^*\in\{1,\cdots,l_0\}$ such that, for $|l-l^*|\leq cL^\beta$,
    and $\gamma\in\mathcal{C}_l$, one has $|u_\gamma|\geq
    e^{-L^{\beta}/6}$.
  \end{enumerate}
\end{Le}
\begin{proof}[Proof of Lemma~\ref{le:12}]
  To prove Lemma~\ref{le:12}, we compare the values of $u$ on
  $\mathcal{C}_l$ and $\mathcal{C}_{l+1}$, that is, the vectors $C_lu$
  and $C_{l+1}u$ given by Lemma~\ref{le:11}.\\
  First, notice that up to an error of size at most
  $e^{-L^{\beta}/3}$, $u$ on $\mathcal{C}_l$ is determined by its
  coefficient $u_{\gamma_l^+}$, or equivalently, by its coefficient
  $u_{\gamma_l^-}$; in particular as
  $\sin(k_l\pi/(n_l+1))=(-1)^{k_l-1}\sin(k_ln_l\pi/(n_l+1))$, the
  representations~\eqref{eq:14} and~\eqref{eq:45} yields
  \begin{equation}
    \label{eq:51}
    \left||u_{\gamma_l^-}|-|u_{\gamma_l^+}|\right|+\left|u_{\gamma_l^+}
      -\alpha^l\sin(k_l\pi/(n_l+1))\right|\leq Ce^{-L^{\beta}/3}.
  \end{equation}
  Notice also that, as $2\leq n_l\leq 2L+1$ is fixed, for $\D
  \rho^*:=\sqrt{\frac{n_l}2+\frac12\cos\left(\frac{2k_l\pi}{n_l+1}\right)}$,
  one has
  \begin{equation}
    \label{eq:22}
    \sup_{1\leq l\leq l_0}\left|\|C_lu\|-\rho_l|\alpha^l|\right|\leq
    Ce^{-L^{\beta}/3}.
  \end{equation}
  To compare the values of $u$ on $\mathcal{C}_l$ and
  $\mathcal{C}_{l+1}$, we use the second equation of~\eqref{eq:32} or,
  equivalently, the eigenvalue equation for $u$ that reads
  (see~\eqref{eq:32})
  \begin{equation}
    \label{eq:44}
    (-\Delta+V_\omega-\overline{E})u=\Delta E\, u+e
  \end{equation}
  where $e=h_1+h_2$ (see~\eqref{eq:32}); hence, $\|e\|\leq C
  \sqrt{L}\,e^{-L^{\beta}/2}$.\\
  We will discuss three cases depending on how far $\gamma^+_l$ and
  $\gamma^-_{l+1}$ are from one another:
  \begin{enumerate}
  \item if dist$(\mathcal{C}_l,\mathcal{C}_{l+1})\geq3$, that is, if
    $\gamma^+_l<\gamma^+_l+1<\gamma^-_{l+1}-1<\gamma^-_{l+1}$: as
    $\{\gamma^+_l+1,\cdots,\gamma^-_{l+1}-1\}\cap[\cup_{l=1}^{l_0}\mathcal{C}_l]
    =\emptyset$, by~\eqref{eq:14}, we know that $|u_n|\leq
    CL^{2-\alpha}$ for $n\in\{\gamma^+_l+1,\cdots,\gamma^-_{l+1}-1\}$.
    The eigenvalue equation~\eqref{eq:44} at the points $\gamma^+_l+1$
    and $\gamma^-_{l+1}-1$ then tells us that
    \begin{equation*}
      |u_{\gamma^+_l}|+|u_{\gamma^-_{l+1}}|\leq Ce^{-L^{\beta}/3}.
    \end{equation*}
    Thus, by~\eqref{eq:51} and~\eqref{eq:22}
    \begin{equation}
      \label{eq:6}
      \|C_l u\|+\|C_{l+1} u\|\leq Ce^{-L^{\beta}/4}.
    \end{equation}
  \item if dist$(\mathcal{C}_l,\mathcal{C}_{l+1})=2$, that is, if
    $\gamma^+_l<\gamma^+_l+1=\gamma^-_{l+1}-1<\gamma^-_{l+1}$: as
    $\gamma^+_l+1\not\in\cup_{l=1}^{l_0}\mathcal{C}_l$,
    by~\eqref{eq:14}, we know that $|u_{\gamma^+_l+1}|\leq C
    L^{2-\alpha}$. Hence, in the same way as above, the eigenvalue
    equation~\eqref{eq:44} at the point $\gamma_l^++1$ tells us that
    \begin{equation*}
      |u_{\gamma^+_l}+u_{\gamma^-_{l+1}}|\leq Ce^{-L^{\beta}/3}.
    \end{equation*}
    Thus, by~\eqref{eq:51} and~\eqref{eq:22}
    \begin{equation}
      \label{eq:56}
      \left|\,\|C_l u\|-\|C_{l+1} u\|\,\right|\leq
      Ce^{-L^{\beta}/4}. 
    \end{equation}
  \item if dist$(\mathcal{C}_l,\mathcal{C}_{l+1})=1$, that is, if
    $\gamma^+_l+1=\gamma^-_{l+1}$: then, the first equation
    in~\eqref{eq:32} and the decomposition~\eqref{eq:40} yield
    \begin{equation*}
      |u_{\gamma^+_l-1}-\Delta E\,u_{\gamma^+_l}|+|u_{\gamma^-_{l+1}+1}-\Delta E
      \,u_{\gamma^-_{l+1}}|\leq Ce^{-L^{\beta}/3}. 
    \end{equation*}
    The eigenvalue equation~\eqref{eq:44} at the points $\gamma^+_l$
    and $\gamma^-_{l+1}$ yields
    \begin{multline*}
      |u_{\gamma^+_l-1}+u_{\gamma^-_{l+1}}+(\omega_{\gamma^+_l}-
      \overline{E}-\Delta E)u_{\gamma^+_l}| \\+
      |u_{\gamma^+_l}+u_{\gamma^-_{l+1}+1}
      +(\omega_{\gamma^-_{l+1}}-\overline{E}-\Delta
      E)u_{\gamma^-_{l+1}}| \leq Ce^{-L^{\beta}/3}.
    \end{multline*}
    Summing these two equations, we obtain
    \begin{equation*}
      |u_{\gamma^-_{l+1}}+(\omega_{\gamma^+_l}-\overline{E})u_{\gamma^+_l}|
      + |u_{\gamma^+_l}+(\omega_{\gamma^-_{l+1}}-\overline{E})u_{\gamma^-_{l+1}}|
      \leq Ce^{-L^{\beta}/3}.
    \end{equation*}
    Then, as the random variables $(\omega_n)_{n\in\Z}$ are bounded,
    using~\eqref{eq:51} and~\eqref{eq:22}, there exists $C>1$ such
    that
    \begin{equation}
      \label{eq:15}
      \frac1C(\|C_l u\|-C e^{-L^{\beta}/4})\leq
      \|C_{l+1}u\|\leq C(\|C_l u\|+e^{-L^{\beta}/4}).
    \end{equation}
  \end{enumerate}
  Notice that~\eqref{eq:56} and~\eqref{eq:6} also imply
  that~\eqref{eq:15} (at the expense of possibly changing the constant
  $C$) also holds in case (1) and case (2). Hence, for $1\leq l,
  l'\leq l_0$, we have
  \begin{equation}
    \label{eq:23}
    C^{-|l'-l|}\|C_{l'} u\|-C^{|l'-l|}e^{-L^{\beta}/4}\leq
    \|C_l u\|\leq C^{|l'-l|}\|C_{l'} u\|+C^{|l'-l|}e^{-L^{\beta}/4}
  \end{equation}
  If case (1) in the above alternative never holds i.e. if for $1\leq
  l\leq l_0$, one has dist$(\mathcal{C}_l,\mathcal{C}_{l+1})\leq2$,
  then, one has $\#\mathcal{C}\geq L/3$.\\
  We know that $\|Cu\|=1+O(e^{-L^{\beta}/3})$. So, for $L$
  sufficiently large, there exists $1\leq l^*\leq l_0$ such that
  \begin{equation*}
    \|C_{l^*} u\|\geq(2\sqrt{\ell_0})^{-1}\geq(4\sqrt{L})^{-1}. 
  \end{equation*}
  Hence, by~\eqref{eq:23}, either of two things occur
  \begin{itemize}
  \item for some $l$, one has $\|C_l u\|\leq e^{-L^{\beta}/5}$, then
    $|l-l^*|\geq \tilde c L^\beta$ for some $\tilde c>0$; thus,
    $l_0\geq 2\tilde c L^\beta$; and for some $0<c<\tilde c$, for
    $|l-l^*|\leq cL^\beta$, one has $\|C_l u\|\geq e^{-L^{\beta}/5}$.
  \item for $1\leq l\leq l_0$, one has $\|C_l u\|\geq
    e^{-L^{\beta}/5}$; then, case (1) never occurs, thus, by the
    observation made above, $\#\mathcal{C}\geq L/3$
  \end{itemize}
  Finally, notice that, by~\eqref{eq:22},~\eqref{eq:51} and the form
  of $u^l$ (see Lemma~\ref{le:11}), $\|C_l u\|\geq e^{-L^{\beta}/5}$
  implies that $|u_n|\geq e^{-L^\beta/6}$ for
  $n\in\mathcal{C}_l$.\\
  This completes the proof of Lemma~\ref{le:12}.
\end{proof}
\noindent We now show that our characterization of $u$, a solution
of~\eqref{eq:32}, imposes very restrictive conditions on the random
variables $(\omega_\gamma)_{-L\leq \gamma\leq L}$.\\
If $\gamma$ is inside one of the connected components of $\mathcal{C}$, say
$\mathcal{C}_l$, that is, if $\{\gamma-1,\gamma,\gamma+1\}\subset\mathcal{C}_l$,
then, by the first equation in~\eqref{eq:32}, we know that
\begin{equation*}
  |u_{\gamma+1}+u_{\gamma-1}-\Delta E u_\gamma|\leq C e^{-L^{\beta}/3}.
\end{equation*}
Plugging this into~\eqref{eq:44}, the eigenvalue equation for $u$, we
get
\begin{equation*}
  |(\omega_\gamma-\overline{E})u_\gamma|\leq e^{-L^{\beta}/4}.
\end{equation*}
Hence, if $\gamma$ belongs to one of the $(\mathcal{C}_l)_l$ singled out in
Lemma~\ref{le:12}, the lower bound for $|u_\gamma|$ given in
Lemma~\ref{le:12} yields
\begin{equation}
  \label{eq:57}
  |\omega_\gamma-\overline{E}|\leq C e^{-L^{\beta}/12}.
\end{equation}
Now, if $n_l>2$, there exists $\gamma\in\mathcal{C}_l$ such that
$\{\gamma-1,\gamma,\gamma+1\}\subset\mathcal{C}_l$. On the other hand,
if $n_l=2$, then, the approximate eigenvalue equation on
$\mathcal{C}_l$ reads
\begin{equation*}
  \left\| \begin{pmatrix} E &\omega_{\gamma_l^-}\\
      \omega_{\gamma_l^+}& E \end{pmatrix}
    \begin{pmatrix}
      u_{\gamma_l^+}\\u_{\gamma_l^-}
    \end{pmatrix}
  \right\|\leq C\sqrt{L}\,e^{-L^{\beta}/2}.
\end{equation*}
So, if $\|C_lu\|\geq e^{-L^\beta/6}$, one has
\begin{equation}
  \label{eq:25}
  |1-(\omega_{\gamma_l^-}-E)(\omega_{\gamma_l^+}-E)|
  \leq Ce^{-L^{\beta}/3}.
\end{equation}
Hence, we see that the random variables must satisfy at least
$cL^\beta$ distinct conditions of the type~\eqref{eq:57}
or~\eqref{eq:25}. As the random variables are supposed to be
independent, identically distributed with a bounded density, these
condition imply that~\eqref{eq:60} can occur with a given partition
$\mathcal{P}$ and $\mathcal{Q}$ with a probability at most,
$\displaystyle e^{-c L^{2\beta}}$ for some $c>0$. As the total number
of partitions is bounded by $2^L$ and as $\beta>1/2$, we obtain that,
$\P$, the probability that~\eqref{eq:60} holds, is bounded
by~\eqref{eq:39}. This completes the proof of Lemma~\ref{le:5}.\qed
\begin{Rem}
  \label{rem:1}
  The estimate~\eqref{eq:39} can be improved as, actually, not all
  partitions are allowed as we saw in the course of the proof.
  Moreover, it is sufficient to assume that the distribution function
  of the random variables be H{\"o}lder continuous for the method to work.
\end{Rem}
\begin{Rem}
  \label{rem:3}
  We now present a natural weaker analogue of point (2) in
  Lemma~\ref{le:5}. Fix $\rho>0$ and define
  \begin{equation*}
    \Delta\mathcal{E}_L^c=\bigcup_{l=0}^L\sigma(-C_l\Delta
    C_l)+[-L^{-\rho},L^{-\rho}].
  \end{equation*}
  then, for $\rho>3$, one has $|\Delta\mathcal{E}_L^c|\leq
  2L^{2-\rho}$, thus,
  \begin{equation*}
    \left|\bigcap_{n\geq1}\bigcup_{L\geq n}\Delta\mathcal{E}_L^c\right|=0
  \end{equation*}
  Define the set of total measure
  \begin{equation*}
    \Delta\mathcal{E}=\R\setminus\left(\bigcap_{N\geq1}\bigcup_{L\geq N}
      \Delta\mathcal{E}_L^c\right).
  \end{equation*}
  Hence, if $E-E'=\Delta E\in\Delta\mathcal{E}$, for $L$ sufficiently
  large, as
  \begin{equation*}
    \inf_{1\leq l\leq L}\text{dist}
    (\Delta E,\sigma(-C_l\Delta C_l))\geq L^{-\rho},
  \end{equation*}
  by the decomposition~\eqref{eq:40}, a solution $u$ to the first
  equation in~\eqref{eq:32} must satisfy $\|u\|\leq L^{-(\nu-\rho)}$
  if $\|h_1\|\leq L^{-\nu}$. Hence, we obtain
  \begin{Le}
    \label{le:13}
    Fix $\nu>4$. For the discrete Anderson model in dimension 1, for
    $E-E'\in\Delta\mathcal{E}$, for $L$ sufficiently large, if
    $E_j(\omega)$ and $E_k(\omega)$ are simple eigenvalues of
    $H_\omega(\Lambda_L)$ such that
    $|E_k(\omega)-E|+|E_j(\omega)-E'|\leq L^{-\nu}$ then
    $\|\nabla_\omega(E_j(\omega)-E_k(\omega))\|_1\geq L^{-\nu}$.
  \end{Le}
  \noindent This can then be used as Lemma~\ref{le:5} is used in the
  proof of Lemma~\ref{le:1} to prove the following variant of the
  decorrelation estimates in dimension 1
  \begin{Le}
    \label{le:14}
    Assume $d=1$. For $\alpha\in(0,1)$ and $E-E'\in\Delta\mathcal{E}$
    s.t. $\{E,E'\}\subset I$, for any $c>0$, there exists $C>0$ such
    that, for $L\geq3$ and $c L^\alpha\leq \ell\leq L^\alpha/c$, one
    has
    \begin{equation*}
      \P\left(\left\{
          \begin{matrix}
            \sigma(H_\omega(\Lambda_\ell))\cap
            (E+L^{-d}(-1,1))\not=\emptyset,\\
            \sigma(H_\omega(\Lambda_\ell))\cap
            (E'+L^{-d}(-1,1))\not=\emptyset
          \end{matrix}
        \right\}\right)\leq C(\ell/L)^{2d}(\log L)^C.
    \end{equation*}    
  \end{Le}
  \noindent Comparing with Lemma~\ref{le:1}, we improved the bound on
  the probability at the expense of reducing the set of validity in
  $(E,E')$.
\end{Rem}
\subsection{Proof of Lemma~\ref{le:8}}
\label{sec:proof-lemma-1}
Pick $(u,v)\in(\R^+)^{2n}$ such that $\|u\|_1=\|v\|_1=1$. At the
expense of exchanging $u$ and $v$, we may assume that
$\|v\|_2\geq\|u\|_2$. Write $u=\alpha v+v^\perp$ where $\langle
v,v^\perp\rangle=0$. Note that, as all the coefficient of both $u$ and
$v$ are non negative, $v^\perp=0$ is equivalent $u=v$. Let us now
assume $u\not=v$ that is $v^\perp\not=0$. One computes
\begin{equation}
  \label{eq:61}
  \|u\|_2^2=\alpha^2\|v\|_2^2+\|v^\perp\|_2^2\text{ and
  }\|u-v\|_2^2=(\alpha-1)^2\|v\|_2^2+\|v^\perp\|_2^2.
\end{equation}
Moreover, as all the coefficients of $v$ are non negative, $v^\perp$
admits at least one negative coefficient. As all the coefficients of
$u$ are non negative, the decomposition $u=\alpha v+v^\perp$ implies
that $\alpha>0$. The first equation in~\eqref{eq:61} and the condition
$\|v\|_2\geq\|u\|_2$ then imply $\alpha\in(0,1)$. Combining this with
$u=\alpha v+v^\perp$ and $\|u\|_1=\|v\|_1=1$ yields
\begin{equation*}
  0<1-\alpha\leq\|v^\perp\|_1.
\end{equation*}
Hence, by the second equation in~\eqref{eq:61} and the Cauchy-Schwartz
inequality, we get
\begin{equation}
  \label{eq:62}
  \frac1{\sqrt{n}}\|u-v\|_1\leq\|u-v\|_2\leq
  \|v\|_2\|v^\perp\|_1+\|v^\perp\|_2\leq 2\sqrt{n}\|v^\perp\|_2.
\end{equation}
For any $(j,k)$, one has
\begin{equation*}
  \left|\begin{matrix}u_j& u_k\\v_j&v_k\end{matrix}
  \right|=\left|\begin{matrix}v^\perp_j&v^\perp_k\\v_j&v_k\end{matrix}
  \right|.
\end{equation*}
As $\langle v,v^\perp\rangle=0$, one computes
\begin{equation*}
  \begin{split}
    \sum_{j,k}\left|\begin{matrix}u_j& u_k\\v_j&v_k\end{matrix}
    \right|^2&=\sum_{j,k}\left((v_jv_k^\perp)^2+
      (v_kv_j^\perp)^2-2v_jv_k^\perp v_kv_j^\perp\right)\\
    &=2\left(\sum_jv_j^2\right)\left(\sum_k(v_k^\perp)^2\right)
    -2\left(\sum_jv_jv_j^\perp\right)\left(\sum_k v_k
      v_k^\perp\right)\\&=2\|v\|_2^2\|v^\perp\|_2^2
    \geq\frac{1}{2n^3}\|u-v\|_1^2.
  \end{split}
\end{equation*}
Thus,
\begin{equation*}
  \max_{j\not=k}\left|\begin{matrix}u_j& u_k\\v_j&v_k\end{matrix}
  \right|^2\geq \frac1{4n^5}\|u-v\|_1^2
\end{equation*}
which completes the proof of Lemma~\ref{le:8}.\qed
\section{The proofs of Theorems~\ref{thr:4} and~\ref{thr:18}}
\label{sec:proofs-theorems}
In~\cite{Ge-Kl:10}, the authors extensively study the distribution of
the energy levels of random systems in the localized phase. Their
results apply also to the discrete Anderson model; in particular, they
provide a proof of Theorems~\ref{thr:4} and~\ref{thr:18} once the
decorrelation estimates obtained in Lemmas~\ref{le:1} and~\ref{le:9}
are known. We provide an alternate proof. The proof in~\cite{Ge-Kl:10}
relies on a construction that also proves Theorem~\ref{thr:3}
(actually a stronger uniform result). Here, we only prove
Theorems~\ref{thr:4} and~\ref{thr:18} independently of the values of
the limits in Theorem~\ref{thr:3}.\\
The localization centers of Proposition~\ref{pro:1} are not defined
uniquely. One can easily check that, under the assumptions of
Proposition~\ref{pro:1}, all the localization centers for a given
eigenvalue or eigenfunction are contained in a disk of radius at most
$C\log L$ (for some $C>0$). To define a unique localization center, we
order the centers lexicographically and let the localization center
associated to the eigenvalue or eigenfunction be
the largest one (i.e. the one most upper left in dimension 2).\\
We prove
\begin{Le}
  \label{le:7}
  Pick $\alpha\in(0,1)$ and $c>0$. Let $\nu$ be defined by
  (Loc). Assume $\ell=\ell(L)$ satisfies $cL^\alpha\leq \ell\leq
  L^\alpha/c$.\\
  If (Loc) (see Proposition~\ref{pro:1}) is satisfied then, for any
  $p>0$ and $\varepsilon>0$, there exists $L_0>0$ such that, for
  $L\geq L_0$, with probability larger than $1-L^{-p}$,
  \begin{enumerate}
  \item if $(E_j)_{1\leq j\leq J}\in I^J$ are eigenvalues of
    $H_{\omega}(\Lambda_L)$ with localization center in
    $\gamma+\Lambda_\ell$, then the operator
    $H_{\omega}(\gamma+\Lambda_{\ell(1+\varepsilon)})$ has $J$
    eigenvalues, say $(\tilde E_j)_{1\leq j\leq J}$, with localization
    center in $\gamma+\Lambda_{\ell(1+\varepsilon/2)}$ and such that
    $\D\sup_{1\leq j\leq J}|E_j-\tilde E_j|\leq
    e^{-\nu\varepsilon\ell/4}$.
  \item if $(E_j)_{1\leq j\leq J}\in I^J$ are eigenvalues of
    $H_\omega(\gamma+\Lambda_{\ell(1+\varepsilon)})$ with localization
    center in $\gamma+\Lambda_\ell$, then the operator
    $H_{\omega}(\Lambda_L)$ has $J$ eigenvalues, say $(\tilde
    E_j)_{1\leq j\leq J}$, with localization center in
    $\gamma+\Lambda_{\ell(1+\varepsilon/2)}$ and such that
    $\D\sup_{1\leq j\leq J}|E_j-\tilde E_j|\leq
    e^{-\nu\varepsilon\ell/4}$.
  \item if $(E_j)_{1\leq j\leq J}\in I^J$ are eigenvalues of
    $H_\omega(\gamma+\Lambda_{\ell(1+\varepsilon)})$ with localization
    center in $\gamma+(\Lambda_{\ell(1+\varepsilon/2)}
    \setminus\Lambda_\ell)$, then there exists $(\beta_j)_{1\leq j\leq
      J}$ such that, for $1\leq j\leq J$, one has
    \begin{itemize}
    \item $\D\beta_j\in\frac{\varepsilon\ell}{16}\Z^d\cap\left[
        \gamma+(\Lambda_{\ell(1+\varepsilon/2)}
        \setminus\Lambda_\ell)\right]$,
    \item the operator
      $H_{\omega}(\beta_j+\Lambda_{\varepsilon\ell/4})$ has an
      eigenvalue, say $\tilde E_j$, satisfying $|E_j-\tilde E_j|\leq
      e^{-\nu\varepsilon\ell/8}$.
    \end{itemize}
  \end{enumerate}
  The number $\nu>0$ is given by (Loc).
\end{Le}
\noindent Similar results can be found in~\cite{Ge-Kl:10}.
\begin{proof}
  With probability at least $1-L^{-p}$, the conclusions of
  Proposition~\ref{pro:1} hold which we assume from now on.\\
  To prove (1), let $(\varphi_j)_{1\leq j\leq J}$ be normalized
  eigenfunctions associated to $(E_j)_{1\leq j\leq J}$. Then, setting
  $\tilde\varphi_j=\car_{\gamma+\Lambda_{\ell(1+\varepsilon)}}\varphi_j$
  and using~\eqref{eq:19} from (Loc) and the assumption that the
  localization center are in $\gamma+\Lambda_{\ell}$, one obtains
  \begin{gather*}
    \left\|\left(\left(\langle\tilde\varphi_j,\tilde\varphi_k
          \rangle_{\ell^2(\gamma+\Lambda_{\ell(1+\varepsilon)})}\right)\right
      )_{\substack{1\leq j\leq J\\1\leq k\leq J}}-
      \text{Id}\right\|\leq J^2e^{-\nu\varepsilon\ell/4},
    \\
    \sup_{1\leq j\leq
      J}\|\car_{\gamma+(\Lambda_{\ell(1+\varepsilon)}\setminus
      \Lambda_{\ell(1+\varepsilon/2)})}
    \tilde\varphi_j\|_{\ell^2(\gamma+\Lambda_{\ell(1+\varepsilon)})}
    \leq
    e^{-\nu\varepsilon\ell/6},\\
    \sup_{1\leq j\leq
      J}\|(H_\omega(\gamma+\Lambda_{\ell(1+\varepsilon)})-E_j)
    \tilde\varphi_j\|_{\ell^2(\gamma+\Lambda_{\ell(1+\varepsilon)})}\leq
    e^{-\nu\varepsilon\ell/4}.
  \end{gather*}
  This immediately yields (1) for $L$ sufficiently large as
  \begin{itemize}
  \item $J\leq (2L+1)^d$ and $cL^\alpha\leq \ell\leq L^\alpha/c$,
  \item at a localization center, the modulus of an eigenfunction is
    at least of order $L^{-d/2}$.
  \end{itemize}
  Points (2) is proved in the same way. We omit further details.\\
  To prove (3), we set
  $\tilde\varphi_j=\car_{\beta_j+\Lambda_{\ell\varepsilon/4}}\varphi_j$
  where $\beta_j$ is the point in $\frac{\varepsilon\ell}{16}\Z^d$
  closest to the localization center of $\varphi_j$. The
  conclusion then follows from the same reasoning as above.\\
  This completes the proof of Lemma~\ref{le:7}.
\end{proof}
\noindent Pick $(U_j)_{1\leq j\leq J}$, $(k_j)_{1\leq j\leq
  J}\in\N^J$, $(U'_j)_{1\leq j\leq J'}$ and $(k'_j)_{1\leq j\leq
  J}\in\N^{J'}$ as in Theorem~\ref{thr:4}. To prove
Theorems~\ref{thr:4} and~\ref{thr:18}, it suffices to
prove~\eqref{eq:18} for $(U_j)_{1\leq j\leq J}$ and $(U'_j)_{1\leq
  j\leq J'}$ non empty compact intervals which
we assume from now on. \\
Pick $L$ and $\ell$ such that $(2L+1)=(2\ell+1)(2\ell'+1)$, $c
L^\alpha\leq \ell\leq L^\alpha/c$ for some $\alpha\in(0,1)$ and
$c>0$. Pick $\varepsilon>0$ small. Partition
$\D\Lambda_L=\bigcup_{|\gamma|\leq\ell'}\Lambda_\ell(\gamma)$ where
$\Lambda_\ell(\gamma)=(2\ell+1)\gamma+\Lambda_\ell$. For
$\Lambda'\subset\Lambda$ and $U\subset\R$, consider the random
variables
\begin{gather*}
  X(E,U,\Lambda,\Lambda'):=
  \begin{cases}
    1&\text{if }H_\omega(\Lambda)\text{ has at least one eigenvalue in}\\
    & E+(\nu(E)|\Lambda|)^{-1}U
    \text{ with localization center in }\Lambda',\\
    0&\text{if not};
  \end{cases}\\\intertext{ if $\Lambda'=\Lambda$, we write
    $X(E,U,\Lambda):=X(E,U,\Lambda,\Lambda)$, and}
  \Sigma(E,U):=\sum_{|\gamma|\leq\ell'}
  X(E,U,\Lambda,\Lambda_\ell(\gamma)),\quad
  \Sigma(E,U,\ell):=\sum_{|\gamma|\leq\ell'}
  X(E,U,\Lambda_\ell(\gamma)).
\end{gather*}
We prove
\begin{Le}
  \label{le:16}
\begin{multline*}
  \left|\mathbb{P}\left(\left\{\omega;\
        \begin{split}
          &\#\{j;\xi_n(E,\omega,\Lambda)\in U_1\}=k_1\\
          &\quad\vdots\hskip3cm\vdots\\
          &\#\{j;\xi_n(E,\omega,\Lambda)\in U_J\}=k_J\\
          &\#\{j;\xi_n(E',\omega,\Lambda)\in U'_1\}=k'_1\\
          &\quad\vdots\hskip3cm\vdots\\
          &\#\{j;\xi_n(E',\omega,\Lambda)\in U_{J'}\}=k_{J'}
        \end{split}
      \right\} \right)-\mathbb{P}\left(\left\{\omega;\
        \begin{split}
          &\Sigma(E,U_1,\ell)=k_1
          \\&\quad\vdots\hskip1.5cm\vdots\\
          &\Sigma(E,U_J,\ell)=k_J\\
          &\Sigma(E',U'_1,\ell)=k'_1
          \\&\quad\vdots\hskip1.5cm\vdots\\
          &\Sigma(E',U_{J'},\ell)=k'_{J'}
        \end{split}
      \right\} \right)\right|\vers_{L\to+\infty}0,
\end{multline*}
\begin{multline*}
  \left|\mathbb{P}\left(\left\{\omega;\
        \begin{split}
          &\#\{j;\xi_n(E,\omega,\Lambda)\in U_1\}=k_1\\
          &\quad\vdots\hskip3cm\vdots\\
          &\#\{j;\xi_n(E,\omega,\Lambda)\in U_J\}=k_J
        \end{split}
      \right\} \right)- \mathbb{P}\left(\left\{\omega;\
        \begin{split}
          &\Sigma(E,U_1,\ell)=k_1,\\
          &\quad\vdots\hskip1.5cm\vdots\\
          &\Sigma(E,U_J,\ell)=k_J
        \end{split} \right\} \right)\right|\vers_{L\to+\infty}0
\end{multline*}
and
\begin{multline*}
  \left|\mathbb{P}\left(\left\{\omega;\
        \begin{split}
          &\#\{j;\xi_n(E',\omega,\Lambda)\in U'_1\}=k'_1\\
          &\quad\vdots\hskip3cm\vdots\\
          &\#\{j;\xi_n(E',\omega,\Lambda)\in U_{J'}\}=k_{J'}
        \end{split}
      \right\} \right)-\mathbb{P}\left(\left\{\omega;\
        \begin{split}
          &\Sigma(E',U'_1,\ell)=k'_1,\\
          &\quad\vdots\hskip1.5cm\vdots\\
          &\Sigma(E',U_{J'},\ell)=k'_{J'}
        \end{split}
      \right\} \right) \right|\vers_{L\to+\infty}0.
\end{multline*}
  
\end{Le}
\begin{proof}
  We first prove
  \begin{Le}
    \label{le:15}
    For any $p>0$ and $\varepsilon>0$, there exists $C>0$ such that,
    for $U$ a compact interval and $L$ sufficiently large, one has
    \begin{equation}
      \label{eq:24}
      \P\left(\left\{\omega;\ \#\{n;\xi_n(E,\omega,\Lambda)\in
          U\}\not= \Sigma(E,U)\right\} \right)\leq C\ell^d
      L^{-d}(|U|+1)^2+L^{-p}.
    \end{equation}
    and
    \begin{equation}
      \label{eq:26}
      \P\left(\left\{\omega;\ \Sigma(E,U)\not=\Sigma(E,U,\ell)\right\}
      \right) \leq L^{-p}+C\,\varepsilon\,|U|.
    \end{equation}
  \end{Le}
  \begin{proof}[Proof of Lemma~\ref{le:15}]
    As $\D\Lambda=\bigcup_{|\gamma|\leq\ell'}\Lambda_\ell(\gamma)$ and
    these sets are two by two disjoint, the quantities
    $\#\{n;\xi_n(E,\omega,\Lambda)\in U\}$ and $\Sigma(E,U)$ differ if
    and only if, for some $|\gamma|\leq\ell'$, $H_\omega(\Lambda)$ has
    at least two eigenvalues in $E+(\nu(E)|\Lambda|)^{-1}U$ with
    localization center in $\Lambda_\ell(\gamma)$. By
    Lemma~\ref{le:7}, this implies that, except on a set of
    probability at most $L^{-p}$,
    $H_\omega((2\ell+1)\gamma+\Lambda_{2\ell})$ has at least two
    eigenvalues in $U+[-e^{-\nu\ell/4},e^{-\nu\ell/4}]$. Thus, by
    Minami's estimate~\eqref{eq:2}, this happens with a probability at
    most $C\ell^{2d}L^{-2d}(|U|+1)^2+L^{-p}$. Summing this estimate
    over all
    the possible $\gamma$'s, we complete the proof of~\eqref{eq:24}.\\
    The proof of~\eqref{eq:26} is split into two steps. Define
    \begin{equation*}
      \Sigma(E,U,\varepsilon)=\sum_{|\gamma|\leq\ell'}
      X(E,U,(2\ell+1)\gamma+\Lambda_{\ell(1+\varepsilon)},
      (2\ell+1)\gamma+\Lambda_{\ell(1-\varepsilon)}).
    \end{equation*}
    Then, we successively prove
    \begin{gather}
      \label{eq:20}
      \P\left(\left\{\omega;\ \Sigma(E,U)\not=
          \Sigma(E,U,\varepsilon)\right\} \right) \leq
      L^{-p}+C\,\varepsilon\,|U| \\\intertext{and} \label{eq:28}
      \P\left(\left\{\omega;\ \Sigma(E,U,\varepsilon)
          \not=\Sigma(E,U,\ell) \right\} \right) \leq
      L^{-p}+C\,\varepsilon\,|U|
    \end{gather}
    which implies~\eqref{eq:26}.\\
    To prove~\eqref{eq:20}, we note that, by Lemma~\ref{le:7}, except
    on a set of probability at most $L^{-p}$, $\Sigma(E,U)$ and
    $\Sigma(E,U,\varepsilon)$ differ if and only if, for some
    $|\gamma|\leq\ell'$, one has
    \begin{enumerate}
    \item either $ \sigma(H_\omega(\Lambda))\cap\delta\tilde
      U\not=\emptyset$,
    \item or $\sigma(H_\omega((2\ell+1)\gamma+
      \Lambda_{\ell(1+\varepsilon)}))\cap \delta\tilde
      U\not=\emptyset$,
    \item or $H_\omega((2\ell+1)\gamma+\Lambda_{\ell(1+\varepsilon)})$
      has an eigenvalue in $\tilde U$ with a localization center in
      the cube $(2\ell+1)\gamma+(\Lambda_{\ell(1+\varepsilon)}
      \setminus\Lambda_{\ell(1-\varepsilon)})$
    \end{enumerate}
    where $\tilde U=E+(\nu(E)|\Lambda|)^{-1}U+
    e^{-\nu\varepsilon\ell/8}[-1,1]$ and
    $\delta\tilde U= \tilde U\setminus (E+(\nu(E)|\Lambda|)^{-1}U)$.\\
    The probability of alternatives (1) and (2) is estimated using the
    Wegner estimate~\eqref{eq:1}. It is bounded by $2L^d
    e^{-\nu\varepsilon\ell/8}\leq L^{-p}$ for $L$ sufficiently
    large.\\
    By point (3) of Lemma~\ref{le:7}, except on a set of probability
    at most $L^{-p}$, alternative (3) implies that, for some
    $\beta\in\gamma+(\Lambda_{\ell(1+\varepsilon/2)}
    \setminus\Lambda_\ell)$, the operator
    $H_{\omega}(\beta+\Lambda_{\varepsilon\ell/4})$ has an eigenvalue
    in
    $E+(\nu(E)|\Lambda|)^{-1}U+e^{-\nu\varepsilon\ell/8}[-1,1]$. The
    number of possible $\beta$'s is bounded by
    $C\varepsilon\ell^d\varepsilon^{-d}\ell^{-d}=C\varepsilon^{1-d}$. Using
    Wegner's estimate~\eqref{eq:1} and summing over the possible
    $\beta$'s, this probability is bounded by
    $C\varepsilon^{1-d}(\varepsilon\ell/L)^d|U|+L^{-p}\leq
    C\varepsilon(\ell/L)^d|U|+L^{-p}$. Finally, we sum this over all
    possible $\gamma$'s to obtain that the probability that
    alternative (3) holds for some $\gamma$ is bounded by
    $C\varepsilon|U|+L^{-p}$.
    This yields~\eqref{eq:20}.\\
    To prove~\eqref{eq:28}, the reasoning is similar. By
    Lemma~\ref{le:7}, except on a set of probability at most $L^{-p}$,
    $\Sigma(E,U,\ell)$ and $\Sigma(E,U,\varepsilon)$ differ if and
    only if, for some $|\gamma|\leq\ell'$, one has
    \begin{enumerate}
    \item either
      $\sigma(H_\omega((2\ell+1)\gamma+\Lambda_{\ell(1+\varepsilon)}))\cap
      \tilde U\not=\emptyset$,
    \item or $\sigma(H_\omega(\Lambda_{\ell}(\gamma)))\cap \tilde
      U\not=\emptyset$,
    \item or $H_\omega((2\ell+1)\gamma+\Lambda_{\ell(1+\varepsilon)})$
      has an eigenvalue in $\tilde U$ with localization center in
      the cube $(2\ell+1)\gamma+(\Lambda_{\ell(1+\varepsilon)}
      \setminus\Lambda_{\ell(1-\varepsilon)})$.
    \item or $H_\omega(\Lambda_{\ell}(\gamma))$ has an eigenvalue in
      $\tilde U$ with localization center in
      $(2\ell+1)\gamma+(\Lambda_{\ell}
      \setminus\Lambda_{\ell(1-\varepsilon)})$.
    \end{enumerate}
    Following the same steps as in the proof of~\eqref{eq:20}, we
    obtain~\eqref{eq:28}. We omit further details.\\
    This completes the proof of Lemma~\ref{le:15}.
  \end{proof}
  \noindent As $\varepsilon>0$ can be chosen arbitrarily small and $J$
  and $J'$ are finite and fixed, Lemma~\ref{le:15} clearly implies
  Lem\-ma~\ref{le:16}.
\end{proof}
\noindent In view of Theorem~\ref{thr:3} and Lemma~\ref{le:16}, to
prove~\eqref{eq:18}, it suffices to prove that, in the limit
$L\to+\infty$, the difference between the following quantities vanishes
%
%
\begin{equation}
  \label{eq:47}
  \mathbb{P}\left(\left\{\omega;\
      \begin{split}
        &\Sigma(E,U_1,\ell)=k_1,\cdots,
        \Sigma(E,U_J,\ell)=k_J\\
        &\Sigma(E',U'_1,\ell)=k'_1,\cdots,
        \Sigma(E',U_{J'},\ell)=k'_{J'}
      \end{split}
    \right\} \right)
\end{equation}
and
\begin{equation}
  \label{eq:27}
  \mathbb{P}\left(\left\{\omega;\
      \begin{split}
        &\Sigma(E,U_1,\ell)=k_1,\\
        &\quad\vdots\hskip1.5cm\vdots\\
        &\Sigma(E,U_J,\ell)=k_J
      \end{split} \right\} \right)\, \mathbb{P}\left(\left\{\omega;\
      \begin{split}
        &\Sigma(E',U'_1,\ell)=k'_1,\\
        &\quad\vdots\hskip1.5cm\vdots\\
        &\Sigma(E',U_{J'},\ell)=k'_{J'}
      \end{split}
    \right\} \right).
\end{equation}
Both terms in~\eqref{eq:47} and~\eqref{eq:27} define probability
measures on $\N^{J+J'}$.  By Theorems~\ref{thr:3} and
Lemma~\ref{le:16}, we know that the limit of the term in~\eqref{eq:27}
also defines a probability measure on $\N^{J+J'}$. Thus, by standard
results on the convergence of probability measures (see
e.g.~\cite{MR1700749}), the difference of~\eqref{eq:47}
and~\eqref{eq:27} vanishes in the limit $L\to+\infty$ if and only if,
for any $(t_j)_{1\leq j\leq J}$ and $(t_{j'})_{1\leq j'\leq J'}$ real,
in the limit $L\to+\infty$, the following quantity vanishes
\begin{multline*}
  \mathbb{E}\left(e^{-\sum_{j=1}^Jt_j\,\Sigma(E,U_j,\ell)-
      \sum_{j'=1}^{J'}t_{j'}\,\Sigma(E',U_{j'},\ell)}\right)\\
  -\mathbb{E}\left(e^{-\sum_{j=1}^Jt_j\,\Sigma(E,U_j,\ell)}\right)\,
  \mathbb{E}\left(e^{-\sum_{j'=1}^{J'}t_{j'}\,\Sigma(E',U_{j'},\ell)}\right).
\end{multline*}
Note that, as the sets $(\Lambda_\ell(\gamma))_{|\gamma|\leq\ell'}$
are two by two disjoint and translates of each other, for a fixed $U$,
the random variables
$(X(E,U,\Lambda_\ell(\gamma))_{|\gamma|\leq\ell'}$ are
i.i.d. Bernoulli random variables. Thus,
\begin{multline*}
  \mathbb{E}\left(e^{-\sum_{j=1}^Jt_j\,\Sigma(E,U_j,\ell)-
      \sum_{j'=1}^{J'}t_{j'}\,\Sigma(E',U_{j'},\ell)}\right)\\
  =\prod_{|\gamma|\leq\ell'}
  \mathbb{E}\left(e^{-\sum_{j=1}^Jt_j\,X(E,U_j,\Lambda_\ell(\gamma))-
      \sum_{j'=1}^{J'}t_{j'}\,X(E',U_{j'},\Lambda_\ell(\gamma))}\right).
\end{multline*}
The Minami estimate~\eqref{eq:2} and the decorrelation
estimates~\eqref{eq:8} and~\eqref{eq:16} of Lemmas~\ref{le:1}
and~\ref{le:9} guarantee that, for any $\rho\in(0,1)$, one has, for
some $C>0$ independent of $\gamma$,
\begin{equation}
  \label{eq:48}
  \begin{split}
    &\sup_{1\leq j<\tilde j\leq J}\mathbb{P}\left(
      \begin{split}
        X(E,U_j,\Lambda_\ell(\gamma))&=1\\
        X(E,U_{\tilde j},\Lambda_\ell(\gamma))&=1
      \end{split}\right)\\
    &\hskip1cm+ \sup_{1\leq j'<\tilde j'\leq J'}\mathbb{P}\left(
      \begin{split}
        X(E',U_{j'},\Lambda_\ell(\gamma))&=1\\
        X(E',U_{\tilde j'},\Lambda_\ell(\gamma))&=1
      \end{split}\right)\\&\hskip2cm
    + \sup_{\substack{1\leq j\leq J\\1\leq j'\leq J'}}
    \mathbb{P}\left(
      \begin{split}
        X(E,U_j,\Lambda_\ell(\gamma))&=1\\
        X(E',U_{j'},\Lambda_\ell(\gamma))&=1
      \end{split}\right)
    \leq C \left(\frac{\ell}{L}\right)^{d(1+\rho)}.
  \end{split}
\end{equation}
Using this, we compute
\begin{equation}
  \label{eq:28}
  \begin{split}
    &\mathbb{E}\left(e^{-\sum_{j=1}^Jt_j\,X(E,U_j,\Lambda_\ell(\gamma))-
        \sum_{j'=1}^{J'}t_{j'}\,X(E',U_{j'},\Lambda_\ell(\gamma))}\right)\\
    &=1+\sum_{j=1}^J(e^{-t_j}-1)\cdot
    \mathbb{P}(X(E,U_j,\Lambda_\ell(\gamma))=1)\\&\hskip.5cm+
    \sum_{j'=1}^{J'}(e^{-t_{j'}}-1)\cdot
    \mathbb{P}(X(E',U_{j'},\Lambda_\ell(\gamma))=1)
    +O\left(\left(\frac{\ell}{L}\right)^{d(1+\rho)}\right).
  \end{split}
\end{equation}
Here, the term $O((\ell/L)^{d(1+\rho)})$ is uniform in $\gamma$.\\
On the other hand, one has
\begin{equation}
  \label{eq:49}
  \begin{split}
    \mathbb{E}\left(e^{-t_j\,X(E,U_j,\Lambda_\ell(\gamma))}\right)
    =1+(e^{-t_j}-1)\cdot \mathbb{P}(X(E,U_j,\Lambda_\ell(\gamma))=1), \\
    \mathbb{E}\left(e^{-t_{j'}\,X(E',U_{j'},\Lambda_\ell(\gamma))}\right)
    =1+(e^{-t_{j'}}-1)\cdot \mathbb{P}(X(E,U_j,\Lambda_\ell(\gamma))=1).
  \end{split}
\end{equation}
By the Wegner estimate~\eqref{eq:1}, we know that
\begin{equation}
  \label{eq:46}
  \sup_{\substack{1\leq j\leq J\\1\leq j'\leq J'}}
  \left[\mathbb{P}(X(E,U_j,\Lambda_\ell(\gamma))=1)+
  \mathbb{P}(X(E',U_{j'},\Lambda_\ell(\gamma))=1)\right]\leq
  C\left(\frac{\ell}{L}\right)^d.    
\end{equation}
Thus, by~\eqref{eq:28}, we have
\begin{multline*}
  \mathbb{E}\left(e^{-\sum_{j=1}^Jt_j\,X(E,U_j,\Lambda_\ell(\gamma))-
      \sum_{j'=1}^{J'}t_{j'}\,X(E',U_{j'},\Lambda_\ell(\gamma))}\right)
  \\=\prod_{j=1}^J\mathbb{E}\left(e^{-t_j\,X(E,U_j,\Lambda_\ell(\gamma))}
  \right)\prod_{j'=1}^{J'}
  \mathbb{E}\left(e^{-t_{j'}\,X(E',U_{j'},\Lambda_\ell(\gamma))}\right)
  \left[1+O\left(\left(\ell L^{-1}\right)^{d(1+\rho)}\right)\right].
\end{multline*}
In the same way, one proves
\begin{equation}
  \label{eq:53}
  \begin{split}
    \mathbb{E}\left(e^{-\sum_{j=1}^Jt_j\,X(E,U_j,\Lambda_\ell(\gamma))}\right)
    =\prod_{j=1}^J\mathbb{E}\left(e^{-t_j\,X(E,U_j,\Lambda_\ell(\gamma))}
    \right)\left[1+O\left(\left(\ell
          L^{-1}\right)^{d(1+\rho)}\right)\right], \\
    \mathbb{E}\left(e^{-\sum_{j'=1}^{J'}t_{j'}\,X(E',U_{j'},\Lambda_\ell(\gamma))}\right)
    =\prod_{j'=1}^{J'}
    \mathbb{E}\left(e^{-t_{j'}\,X(E',U_{j'},\Lambda_\ell(\gamma))}\right)
    \left[1+O\left(\left(\ell L^{-1}\right)^{d(1+\rho)}\right)\right].
  \end{split}
\end{equation}
As $\#\{|\gamma|\leq\ell'\}\leq C(L\ell^{-1})^d$, we obtain that
\begin{multline}
  \label{eq:55}
  \mathbb{E}\left(e^{-\sum_{j=1}^Jt_j\,\Sigma(E,U_j,\ell)-
      \sum_{j'=1}^{J'}t_{j'}\,\Sigma(E',U_{j'},\ell)}\right)\\
  =\mathbb{E}\left(e^{-\sum_{j=1}^Jt_j\,\Sigma(E,U_j,\ell)}\right)\,
  \mathbb{E}\left(e^{-\sum_{j'=1}^{J'}t_{j'}\,\Sigma(E',U_{j'},\ell)}\right)
  \left[1+O\left(\left(\ell L^{-1}\right)^{d\rho}\right)\right].
\end{multline}
Finally, note that~\eqref{eq:49},~\eqref{eq:46} and~\eqref{eq:53}
imply that, for any $(t_j)_{1\leq j\leq J}$ and $(t_{j'})_{1\leq
  j'\leq J'}$, one has
\begin{equation*}
  \sup_{L\geq1}\left[
  \mathbb{E}\left(e^{-\sum_{j=1}^Jt_j\,\Sigma(E,U_j,\ell)}\right)+
  \mathbb{E}\left(e^{-\sum_{j'=1}^{J'}t_{j'}\,\Sigma(E',U_{j'},\ell)}\right)
  \right]<+\infty.
\end{equation*}
Hence, by~\eqref{eq:55}, as $c L^\alpha\leq \ell\leq L^\alpha/c$ for
some $\alpha\in(0,1)$, we obtain that
\begin{multline*}
  \mathbb{E}\left(e^{-\sum_{j=1}^Jt_j\,\Sigma(E,U_j,\ell)-
      \sum_{j'=1}^{J'}t_{j'}\,\Sigma(E',U_{j'},\ell)}\right)\\
  -\mathbb{E}\left(e^{-\sum_{j=1}^Jt_j\,\Sigma(E,U_j,\ell)}\right)\,
  \mathbb{E}\left(e^{-\sum_{j'=1}^{J'}t_{j'}\,\Sigma(E',U_{j'},\ell)}\right)
  \vers_{L\to+\infty}0.
\end{multline*}
This completes the proof of Theorems~\ref{thr:4} and~\ref{thr:18}.
\begin{Rem}
  \label{rem:4}
  The basic idea we used here is to split the cube $\Lambda$ into
  smaller two-by-two disjoint cubes $(\Lambda_\gamma(\ell))_\gamma$ in
  such a way that, up to exponentially small errors, the eigenvalues
  of $H_\omega(\Lambda)$ can be represented as eigenvalues for
  $H_\omega(\Lambda_\gamma(\ell))$ and that they are independent of
  each other. In~\cite{Ge-Kl:10} (see also~\cite{Ge-Kl:10b} for a
  review of the results), this idea is exploited thoroughly to study
  the eigenvalue statistics for random operators in the localized
  regime.
\end{Rem}
\section{Proof of Proposition~\ref{pro:1}}
\label{sec:proof-proposition}
Let $I$ be a compact subset of the region of localization i.e. the
region of $\Sigma$ where the finite volume fractional moment criteria
of~\cite{MR2002h:82051} for $H_\omega(\Lambda)$ are verified for
$\Lambda$ sufficiently large. Then, by (A.6) of~\cite{MR2002h:82051},
we know that there exists $\alpha>0$ such that, for any $F\subset I$,
$\forall(x,y)\subset\Lambda^2$, one has
\begin{equation}
  \label{eq:3}
  \esp(|\mu_{\omega,\Lambda}^{x,y}|(F))\leq C e^{-\alpha|x-y|}.
\end{equation}
where $\mu_{\omega,\Lambda}^{x,y}$ denotes the spectral measures of
$H_\omega(\Lambda)$ associated to the vector $\delta_x$ and
$\delta_y$. In particular, if $F$ contains a single eigenvalue of
$H_\omega(\Lambda)$, say $E$, that is simple and associated to the
normalized eigenvector, say, $\varphi$ then
\begin{equation}
  \label{eq:4}
  |\mu_{\omega,\Lambda}^{x,y}|(F)=|\varphi(x)|\,|\varphi(y)|.
\end{equation}
Pick $\varepsilon$ and $\delta$ positive such that
$\varepsilon|\Lambda|^2=\delta/K$ for some large $K$ to be chosen
below. Then, partition $I=\cup_{1\leq n\leq N}I_n$ into intervals
$(I_n)_n$ of length $\varepsilon$. By Minami's estimate, one has
\begin{equation*}
  \P(\{\omega;\ \exists n\text{ s.t. }I_n\text{ contains 2 e.v. of
  }H_\omega(\Lambda)\})\leq C\delta|I|/K\leq \delta/2
\end{equation*}
if $C|I|/K\leq1/2$. Pick $K$ so that this be satisfied.\\
We now apply~\eqref{eq:3} to $F=I_n$ for $1\leq n\leq N$ and sum the
results for $s<\alpha$ to get
\begin{equation*}
  \forall y\in\Lambda,\quad \esp\left(\sum_n\sum_{x\in\Lambda}
    e^{s|x-y|}|\mu_{\omega,\Lambda}^{x,y}|(I_n)\right)\leq
  C|I|\varepsilon^{-1}.
\end{equation*}
Hence, by Markov's inequality,
\begin{equation*}
  \P\left(\sum_n\sum_{(x,y)\in\Lambda^2}
    e^{s|x-y|}|\mu_{\omega,\Lambda}^{x,y}|(I_n)\geq
    \frac{\tilde C|\Lambda|\,|I|}{\delta\varepsilon}\right)
  \leq \delta/2.
\end{equation*}
Thus, using the relation between $\delta$ and $\varepsilon$, with a
probability larger than $1-\delta$, we know that
\begin{enumerate}
\item each interval $I_n$ contain at most a single eigenvalue, say,
  $E_n$ associated to the normalized eigenfunction, say, $\varphi_n$;
\item by~\eqref{eq:4}, one has
  \begin{equation*}
    \forall(x,y)\in\Lambda^2,\quad
    |\varphi_n(x)|\,|\varphi_n(y)|\leq
    \frac{C|\Lambda|^3e^{-s|x-y|}}{\delta^2}.
  \end{equation*}
\end{enumerate}
As $\varphi_n$ is normalized, if $x_n$ is a maximum of
$x\mapsto|\varphi_n(x)|$, one has
\begin{equation*}
  |\varphi_n(x_n)|\geq|\Lambda|^{-1/2},
\end{equation*}
thus,
\begin{equation*}
  \forall x\in\Lambda,\quad
  |\varphi_n(x)|\leq
  C|\Lambda|^{7/2}\delta^{-2}e^{-s|x-x_n|}.
\end{equation*}
This yields Proposition~\ref{pro:1} if one picks $\delta=L^{-p}$ when
$\Lambda=\Lambda_L$.\qed
%
%
\def\cprime{$'$} \def\cydot{\leavevmode\raise.4ex\hbox{.}} \def\cprime{$'$}

\end{document}